\documentclass[sigconf,nonacm]{acmart} 

\usepackage{preamble}

\usepackage{balance} 

\title[Optimal Voting]{Optimal Voting: Beyond a Utilitarian View}

\author{Ben Armstrong}
\affiliation{
  \institution{Tulane University}
  \city{New Orleans}
  \country{United States}}
\email{barmstrong@tulane.edu}

\author{Nicholas Mattei}
\affiliation{
  \institution{Tulane University}
  \city{New Orleans}
  \country{United States}}
\email{nmattei@tulane.edu}

\begin{abstract}
A common measure of the quality of a voting rule is the social welfare of the alternatives it elects. Prior work shows that, when voter utilities are drawn from a neutral distribution, expected utilitarian welfare is maximized by a positional scoring rule. We extend this approach to the power mean family of welfare functions, which includes utilitarian, Nash, and egalitarian welfare as special cases. We show that, across the power mean family, a deterministic positional scoring rule (PSR) is asymptotically optimal for expected social welfare among all social choice functions and we characterize this rule for a broad class of utility distributions. We further show that, among randomized PSRs, plurality uniquely maximizes expected social welfare for all power means. Experimentally, we find that optimal score vectors often yield substantially higher welfare than canonical score vectors. On highly structured utility distributions, optimal score vectors take on novel shapes unlike existing canonical score vectors, highlighting open questions for future analysis.
\end{abstract}


\begin{document}




\maketitle 


\section{Introduction}
\label{sec:introduction}

Voting is used to aggregate the preferences of a group of stakeholders over a set of alternatives in order to arrive at a shared group decision \cite{brandt2016handbook,Zwicker:Voting}.
Voting rules frequently operate on a ranking over alternatives provided by each voter, yet voters often have an underlying intensity of preference or \textit{utility} associated with each alternative.
Voters, who may be part of a human society or a digital agent such as a recommender system~\cite{aird2024dynamic} or AI agent, may not even know their exact utility over alternatives; this might arise from uncertain or changing human preferences, lack of alignment between human preferences and AI agent representatives~\cite{conitzer2024social}, or even noisiness in a communication environment. As a result, ordinal rankings are relied upon to achieve good outcomes from a voting procedure. The quality of such a procedure might be measured in two ways: First, voting rules may be desirable due to the axiomatic properties which they satisfy. Second, voting rules aim to optimize some target function based on voter preferences.

Under the axiomatic perspective, much work studies voting procedures which provably satisfy some set of desirable principles for all possible preferences~\cite{Zwicker:Voting}. For instance, a system designer might wish to ensure that a single voter does not have unilateral control over election outcomes (non-dictatorship) or that an alternative who is preferred head-to-head against all other alternatives is always selected (Condorcet consistency).
The optimization view, which we focus on in this paper, sees the goal of a voting rule as finding some outcome that maximizes a target function. This may be, for example, a straightforward sum of voter utilities --- known as utilitarian welfare. A frequent focus of research is on \textit{worst-case} analysis; finding the voting rule which maximizes the worst possible outcome that a voting rule might produce \cite{procaccia2006distortion}. 

Recently, the optimization view has been shown effective on axiomatic settings using black-box techniques, i.e., machine learning, to find novel rules that rarely violate axioms, rather than provably satisfying axioms \cite{caiata2026voting,matone2024deepvoting}. These approaches show the existence of such rules but are not immediately useful in practice. An optimization approach constrained to a bounded class of rules, such as score based rules, has shown that the axiomatic approach can be cleanly combined with optimization to blend these schools of thought~\cite{berker2026designing}.

\textit{Contribution.} 
We focus on the optimization task of maximizing a broad class of function over voter utilities: the \textit{power mean}, which includes special cases such as utilitarian (utility sum), Nash (utility product), and egalitarian (min utility) welfare functions. We find that deterministic positional scoring rules (PSRs) asymptotically maximize welfare and explore how the optimal rule changes based on utility distribution and power mean. These rules are particularly desirable as an optimization target as they are well understood axiomatically~\cite{young1975social,skowron2019axiomatic}, thus blending both schools of voting. 
This work generalizes and extends prior results for specific instances of the power mean family \cite{apesteguia2011justice,boutilier2015optimal}.
Among our results, we introduce a novel description of the rule exactly optimizing Nash social welfare and strengthen previous results on egalitarian welfare. We also show that the best randomized PSR is always a $k$-approval rule and, under a minimal assumption on voter independence, is plurality. These contributions are summarized in \Cref{tab:our_contribution}. Experimental analysis supports these results and finds the optimized positional scoring vector for distributions outside of our theoretical analysis, highlighting open questions for future research and suggesting that the well-known Borda rule may prove to maximize social welfare on certain highly structured preference domains.

\section{Related Work}

Much of the past work is focused on a worst-case analysis and many papers have explored the optimization of utilitarian social welfare.  We find less work that studies other social welfare functions.

\textit{Average Case Analysis}
Several previous works have studied the rule which maximizes utilitarian social welfare and few that study other functions. \citet{weber1978reproducing} assumes voter utilities are distributed independently and uniformly on an interval and show that Borda asymptotically maximizes utilitarian welfare. \citet{bordley1983pragmatic} and \citet{merrill1984comparison} use computer simulation to expand these results to additional utility distributions and voting rules in the utilitarian setting. Borda was shown to be the optimal score vector for expected utilitarian welfare by \citet{apesteguia2011justice}, when utilities are sampled i.i.d. from a uniform distribution for all voters. Under a slightly different model, \citet{boutilier2015optimal} also show that expected utilitarian welfare is maximized by a scoring rule based on the order statistics of the underlying utility distribution and, on uniform utilities, is the Borda vector for any number of voters.  Later work shows a novel positional scoring rule, \textit{binomial voting}, approximately maximizes utilitarian welfare \cite{gonczarowski2023distortion}. ~\citet{pivato2016asymptotic} shows Borda is highly likely to elect the alternative maximizing social welfare in a given profile, under varying utility distributions. Experimental analysis also shows that on Euclidean preferences, Borda typically results in much higher utilitarian welfare than other voting rules \cite{green2016statistical}.

\textit{Worst Case Analysis}
Early work from \citet{procaccia2006distortion} introduced distortion as a measure of the worst-case quality of voting rules and shows that many deterministic PSRs, such as Borda and Veto, have unbounded distortion. Subsequent work shows that despite poor worst-case welfare, all PSRs can be learned with an efficient number of samples in polynomial time, highlighting the possibility of attempting to learn novel PSRs optimized for a metric of interest, such as expected welfare~\cite{procaccia2009learnability}.
\citet{boutilier2015optimal} show that a randomized rule based on the \textit{harmonic} score vector $(1, \nicefrac{1}{2}, \nicefrac{1}{3}, \dots)$ is near-optimal for utilitarian distortion.



\begin{table}[tbp]
\centering
\begin{tabular}{|c|c|c|}
\hline
\textbf{Welfare Function}     & \textbf{Deterministic}                            & \textbf{Randomized}                     \\ \hline
Power Mean                    & $\mathbb{E}[\psi_p(U_{(k)})]$                     & \multirow{5}{*}{$k$-Approval$^\dagger$} \\ \cline{1-2}
Egalitarian ($p \to -\infty$) & Veto$^*$             &                                         \\ \cline{1-2}
Nash($p = 0$)                 & $n \log \mathbb{E}[U_{(k)}^{1/n}]$                &                                         \\ \cline{1-2}
Utilitarian($p = 1$)          & $\mathbb{E}[U_{(k)}]$ \cite{boutilier2015optimal} &                                         \\ \cline{1-2}
Maximax ($p \to \infty$)      & Plurality$^*$        &                                         \\ \hline
\end{tabular}
\caption{Average-case welfare maximizing PSRs. Results with no citation are new. We show deterministic rules asymptotically maximize power mean welfare across all social choice functions. $*$: \citet{apesteguia2011justice} provide non-asymptotic, non-PSRs exactly maximizing welfare. $\dagger$: With neutrality on alternatives and independence among voters this is plurality.}
\label{tab:our_contribution}
\end{table}

\textit{Optimization of Non-Utilitarian SW} 
In addition to their utilitarian result, \citet{apesteguia2011justice} also show that egalitarian and maximax welfare are \textit{not} optimized by a PSR as social welfare cannot be exactly represented as a summation over per-voter terms. Use of the power mean function as a welfare function in voting was first proposed by \citet{cousins2022uncertainty} who apply a sample-complexity approach to identify possible causes of systemic bias. This is extended to the \textit{weighted} power mean function by \citet{pardeshi2024learning} who study the PAC-learnability of welfare functions based on their weights. A similar welfare function has also been used in studying fair allocations \cite{barman2020tight,barman2022universal}. The optimization of arbitrary utility functions has been discussed by \citet{armstrong2026optimizing} and applied to a measure of axiomatic satisfaction \cite{caiata2026voting}. Nash-based rules have been studied normatively for allocations \cite{airiau2023portioning} and voting \cite{aziz2019fair}. \citet{aziz2019fair} detail properties of a Nash rule which takes as input voter utilities but do not comment on ordinal rules. \citet{ebadian2024optimized} study Nash welfare and show that a rule exists which can always achieve within a $\Theta(\log m)$ approximation to the best welfare.

\section{Model}
\label{sec:model}

Following a typical model of voting \cite{Zwicker:Voting}, we have a set of alternatives $A = \{a_1, a_2, ..., a_m\}$ and a set of voters $V = \{v_1, v_2, ..., v_n\}$. We assume each voter $v_i$ draws a utility vector $(u_i(a))_{a \in A}$ from a distribution $D$ over $\mathbb{R}_{>0}^m$, giving a utility function $u_i: A \rightarrow \mathbb{R}_{>0}$ where $u_i(a)$ denotes the utility to $v_i$ if $a$ is selected as a winner. The set of utility functions indexed by voter is called a \textit{utility profile} $\mathbf{u}$.

Each utility function induces a preference order $\sigma$ over the alternatives $A$, where $\sigma(a)$ refers to the rank of $a$ in $\sigma$ and the collection of all preference orders is a preference profile $\boldsymbol{\sigma}$. For a voter $v_i$, $\sigma_i$ denotes a strict ordering over $A$ such that $\sigma_i(a) < \sigma_i(a') \iff u_i(a) > u_i(a')$. That is, $v_i$ prefers $a$ over $a'$ if, and only if, $v_i$ receives more utility from the election of $a$ than from the election of $a'$. Let $\boldsymbol{\sigma}^\mathbf{u}$ refer to the preference profile containing all orderings induced by a utility profile $\mathbf{u}$. The sets of all possible utility and preference profiles are denoted, respectively, $\mathbb{U}$ and $\mathbb{P}$. We use the term ``profile'' to refer to both types of profile as contextually appropriate.

We introduce three utility assumptions used throughout.


\begin{itemize}

    \asmitem{A1}{a1} Voters draw their utility vectors independently from the same neutral distribution $D$ on $(0, \infty)^m$. $D$ is unchanged when the alternatives are relabelled and the joint distribution of utility profiles is anonymous and neutral.
    
    \asmitem{A1$^{+}$}{a1plus} Each voter's utilities for the $m$ alternatives are i.i.d.\ draws from a distribution with CDF $F$, shared between all voters.
    
    \asmitem{A2}{a2} No voter assigns equal utility to two alternatives and each ranking $\sigma_i$ is strict.
\end{itemize}

Under \Cref{a1,a2} our model matches the utility assumptions used by \citet{boutilier2015optimal}. If we use \Cref{a1plus,a2} we instead get the model of \citet{apesteguia2011justice}.  \Cref{a1} allows utilities over alternatives to be correlated within a single voter's preference order while under \Cref{a1plus} does not. Unless stated otherwise we assume \Cref{a1,a2}.

\subsection{Positional Scoring Rules}
A \textit{social choice function} (SCF) $f: \mathbb{P} \rightarrow 2^A \setminus {\emptyset}$ maps a preference profile to a non-empty set of winning alternatives. An SCF is called \textit{resolute} if it always returns a single winning alternative \cite{Zwicker:Voting}.

We focus on a special class of SCFs known as positional scoring rules (PSR) which are parameterized by a score vector $\mathbf{s} = (s_1, \dots, s_m)$. PSRs assign a score to each alternative based on the number of times it is ranked in each position: $S_\mathbf{s}(a, \boldsymbol{\sigma})$ = $\sum_{\sigma \in \boldsymbol{\sigma}} \mathbf{s}_{\sigma(a)}$. We assume scores are always non-negative.
We consider both deterministic and randomized PSRs. In the deterministic setting winning alternatives are exactly those who receive the maximal score, $f_\mathbf{s}(\boldsymbol{\sigma}) = \{ a \in A  \mid S_\mathbf{s}(a, \boldsymbol{\sigma}) = \max_{a' \in A} S_\mathbf{s}(a', \boldsymbol{\sigma}) \}$.
In the randomized setting (denoted $f_\mathbf{s}^\text{R}$), a single winner is chosen with probability directly proportional to their score; 
$P[f_\mathbf{s}^\text{R}(\boldsymbol{\sigma}) = a] = S_\mathbf{s}(a, \boldsymbol{\sigma}) / (n \bar s)$, where $\bar s := \sum_{j=1}^m s_j$
Observe that, in the deterministic setting, scoring vectors can be scaled additively and multiplicatively without changing the outcome of the corresponding PSR, and can be scaled multiplicatively in the randomized setting. Any rule we refer to as ``unique'' is unique only up to these transformations.

We say that a PSR $f_\mathbf{s}$ is \textit{feasible} if $s_i \geq s_{i+1} \, \forall \, 1 \leq i < m$ and $s_1 > s_m$. For simplicity we typically denote feasible deterministic PSRs as $\mathbf{s} = (1, \dots, 0)$.
In the literature there are a number of classic PSRs that are considered canonical such as Plurality: $(1, 0, \dots, 0)$; Veto: $(1, 1,  \dots, 0)$; Borda: $(m-1, m-2,  \dots, 0)$; Harmonic \cite{boutilier2015optimal}: $\mathbf{s}^\text{h} = (1, \nicefrac{1}{2}, \nicefrac{1}{3}, \dots, \nicefrac{1}{m})$, and Binomial \cite{gonczarowski2023distortion}: $s^\text{bi}_j = \sum_{k=j}^m \binom{m}{k}$. All of these PSRs, visualized in \Cref{fig:canonical_vectors}, target normative ideas of a \textit{correct} score for each rank. e.g., Plurality models a winner-take-all mentality, while Borda linearly scales reward for candidates with their rank.

We consider only resolute rules here and can transform any non-resolute social choice function into a resolute SCF with a \textit{tie-breaking} method. Two tie-breaking methods are common: lexicographic, and random \cite{aziz2013ties,freeman2015general}. Lexicographic tie-breaking requires an ordering over alternatives and would violate rule neutrality so when ties occur we select a winner uniformly at random. 

\begin{figure}[tbp]
    \centering
    \includegraphics[width=0.75\linewidth]{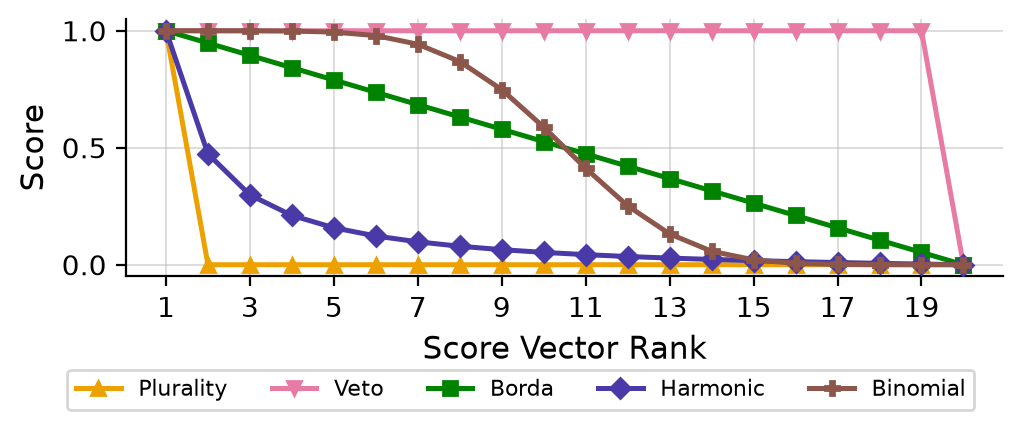}
    \caption{The canonical score vectors discussed in \Cref{sec:model}.
    }
    \label{fig:canonical_vectors}
\end{figure}

\subsection{Social Welfare}

A \textit{social welfare function} (SWF) is a function $W: \mathbb{R}^n \times A \rightarrow \mathbb{R}$ which uses the utility each voter would receive from the election of alternative $a$ to compute a single numeric value representing the quality of $a$ as a social decision. Different SWFs aim to capture different aspects of quality and may be more or less appropriate in varying settings \cite{Shoham:MultiagentSystems}.
We write $W(\{u_i(a)\ \forall\ i\ \in\ V\}, a)$ and assume that $W$ is computable in polynomial time. Due to the range of fairness concerns we focus on a general family of social welfare functions, the power mean, which generalizes existing welfare functions into a single parameterized model~\cite{pardeshi2024learning}, defined (slightly abusing notation) as,

\[
W^\text{PM}(\mathbf{u}, a, p) = 
\begin{cases*}
    \left(\frac{1}{n} \sum_{i=1}^n u_i(a)^p\right)^\frac{1}{p} & if  $p \neq 0$  \\
    (\prod_{i=1}^n u_i(a))^{\frac{1}{n}} & if  $p = 0$  \\
\end{cases*}
\]

Based on the parameter $p \in [-\infty, \infty]$, $W^\text{PM}$ can reduce to a number of special cases from the literature \cite{boutilier2015optimal,apesteguia2011justice}. We treat $p = \pm \infty$ as the $\min$ and $\max$ functions which are approached in the limit.

\begin{enumerate}
    \item ($p = 1$): $\frac{1}{n}\sum_{i \in V} u_i(a)$ measures the mean utility across all voters, i.e., \textit{utilitarian} social welfare. This can be maximized when a single voter receives a large amount of utility, or when all voters receive increased utility.
    \item ($p = 0$):  $\prod_{i \in V} u_i(a)^\frac{1}{n}$ measures the geometric mean of all voter's utility and is maximized when voter utilities are both high \textit{and} equal. Corresponds directly with \textit{Nash} welfare.
    \item ($p \rightarrow -\infty$): $\text{min}_{i \in V} (u_i(a))$ is based on only a single voter's utility, and is maximized when the voter receiving the lowest amount of utility receives as high utility as possible. This is referred to as \textit{Egalitarian} or Rawlsian social welfare.
    \item ($p \rightarrow \infty$):  $\max_{i \in V} u_i(a)$ is also based on a single voter's utility but is maximized when the voter receiving the most utility receives high utility. This function, referred to as \textit{Maximax}, largely represents \textit{unfairness} and is included primarily as a special case of the power mean function.
\end{enumerate}

\subsubsection{Social Welfare Optimization}

We now discuss our optimization framework. We focus on analyzing \textit{average-case} social welfare and study voting rules which maximize this asymptotically. We first introduce assumptions used in our theoretical results.

\textit{Assumptions.}
\label{sec:assumptions}
Let $U$ denote a voter's utility for a fixed alternative (by \Cref{a1} voter and alternative identity do not matter here). Let $U_{(1)} > U_{(2)} > \dots > U_{(m)}$ denote one voter's $m$ utilities in \emph{descending} order (strict by \Cref{a2}), so that $U_{(k)}$ is the utility of the alternative that voter ranks in position $k$. We use this non-standard notation for order statistics so that $U_{(i)}$ corresponds with the indexing of $\mathbf{s}_i$ in a score vector. For fixed $p \in \mathbb{R}$, $\psi_p: (0, \infty) \to \mathbb{R}$ is the Box-Cox transform,

\begin{equation}
\label{eq:psi}
\psi_p(t) = \begin{cases*}
    \frac{t^p - 1}{p} & if  $p \neq 0$  \\
    \log t & if  $p = 0$  \\
\end{cases*}
\end{equation}
\noindent
which \Cref{lem:hp_appendix} shows is continuous,  strictly increasing, and satisfies $\psi_p(W^\text{PM}(\mathbf{u}, a, p)) = \frac{1}{n}\sum_{i=1}^n \psi_p(u_i(a))$.
Many of our proofs rely on the assumption below which holds
for every $p$ whenever the support of $D$ lies in $[\varepsilon, M]$ for some $0 < \varepsilon < M < \infty$.
 
\begin{itemize}
    \asmitem{A3}{a3} For fixed $p$, $\mathbb{E}[\psi_p(U)^2] < \infty$ and $\mathbb{E}[U^q] < \infty$ for some $q > \max\{2, p\}$.
\end{itemize}

\begin{definition}[Expected Social Welfare]
\label{def:esw}
Fix $p \in \mathbb{R} \cup \{\pm \infty\}$ and let $\mathbf{u}$ be a random utility profile of $n$ voters, with induced preference profile $\boldsymbol{\sigma}$. The \emph{expected social welfare} of a (possibly randomized) resolute social choice function $f$ is
\[
\mathrm{ESW}_p(f) = \mathbb{E}\big[W^\text{PM}(\mathbf{u}, f(\boldsymbol{\sigma}), p)\big],
\]
where the expectation is over $\mathbf{u}$ and randomness of $f$ (tie-breaking or a randomized winner), which is independent of $\mathbf{u}$ given $\boldsymbol{\sigma}$.
\end{definition}

By \Cref{a1} every preference profile occurs with positive probability, so $f$ maximizes $\mathrm{ESW}_p$ over all social choice functions if, and only if, at every preference profile $\boldsymbol{\sigma}$ it selects an alternative in $\operatorname*{arg\,max}_{a \in A} \mathbb{E}[W^\text{PM}(\mathbf{u}, a, p) \mid \boldsymbol{\sigma}]$. This matches the average-case optimality definition of \citet{boutilier2015optimal}.

\begin{definition}[Average-Case Optimal]
\label{def:opt}
Let $\mathrm{OPT}_n = \max_g \mathrm{ESW}_p(g)$, where the maximum is over all (possibly randomized) SCFs $g$ on $n$ voters. We say that $f$ is \emph{asymptotically average-case optimal} if

\[
\lim_{n \to \infty} \frac{\mathrm{OPT}_n - \mathrm{ESW}_p(f)}{\mathrm{OPT}_n - \mathrm{ESW}_p(f_\mathbf{1})} = 0.
\]
\noindent
where $\mathbf{1} = (1, \dots, 1)$ and $f_\mathbf{1}$ corresponds to uniform random selection of a winning alternative.
That is, in the limit, the welfare gap between the optimal rule and $f$ shrinks more quickly than the gap between the optimal rule and welfare resulting from a uniform random selection. 
By \Cref{a1}, every alternative approaches the same expected welfare as $n$ grows, so the ESW of every rule converges to the same limit. The ratio magnifies the difference between rules as $n$ grows, so suboptimal rules stay positive in the limit.
\end{definition}


\section{Optimal Deterministic Rules}

This section examines deterministic PSRs -- the winning alternative(s) are exactly those which receive maximal points. We assume any ties among winners are broken uniformly at random. We first show that, for finite $p$, the voting rule that maximizes expected power mean social welfare belongs to the family of PSRs. This follows from \citet{young1975social} who identifies a basic set of axioms which characterize the family of scoring correspondences, and extends the proof of \citet{boutilier2015optimal} around the setting of utilitarian social welfare. Here, we identify the score vectors which maximize expected social welfare for \textit{any} parameterization of the power mean family. The optimal score vector is based on the order statistic of $D$.

\begin{theorem}\label{thm:optimal_is_psr}
Fix $p\in\mathbb{R}$ and assume \Cref{a3}. There is a score vector $\mathbf{s} \in \mathbb{R}^m$ such that the positional scoring rule $f_\mathbf{s}$ is asymptotically average-case optimal. Additionally, when $p = 1$, $f_\mathbf{s}$ maximizes $\mathrm{ESW}_1$ over all social choice functions for every $n$.
\end{theorem}
 
\Cref{thm:optimal_is_psr} applies Young's characterization (\Cref{lem:young_appendix}) of scoring rules to the objective $\mathbb{E}[\sum_{i=1}^n \psi_p(u_i(a)) \mid \boldsymbol{\sigma}]$, which is additive across voters, and shows that this objective determines $\mathrm{ESW}_p$ up to an additive $o(n^{-\frac{1}{2}})$ (\Cref{lem:expansion}). All missing proofs are in \Cref{app:theorems}.

\subsection{The Shape of Optimality}

Our primary theoretical result asymptotically describes the structure of the average-case optimal scoring vector for neutral utility distributions across the family of power mean social welfare functions. Each result below states which of the assumptions from \Cref{sec:assumptions} it relies on, and the optimal score vector is expressed through the Box-Cox transform $\psi_p$ of \Cref{eq:psi}.

\begin{theorem}
\label{thm:optimal_deterministic}
Fix $p \in \mathbb{R}$, assume \Cref{a3}, and let $s_k^*(p) = \mathbb{E}[\psi_p(U_{(k)})]$ for $k = 1, \dots, m$. For every score vector $\mathbf{s} \in \mathbb{R}^m$ the limit
\[
\Delta_p(\mathbf{s}) = \lim_{n \to \infty} n^{\frac{1}{2}}\big(\mathrm{ESW}_p(f_{\mathbf{s}^*(p)}) - \mathrm{ESW}_p(f_\mathbf{s})\big)
\]
exists and is nonnegative, and $\Delta_p(\mathbf{s}) = 0$ if, and only if, $\mathbf{s} = \mathbf{s}^*(p)$. That is, $\mathbf{s}^*(p)$ is the unique asymptotically average-case optimal PSR.
\end{theorem}
 
Together, \Cref{thm:optimal_is_psr,thm:optimal_deterministic} show that $\mathrm{ESW}_p$ is asymptotically maximized by a positional scoring rule, and that this rule is $\mathbf{s}^*(p)$.

\begin{remark}[Asymptotic Results and Limit on p]
\label{rem:exact_vs_asymptotic}
Prior work finds an average case optimal score vector when $p=1$ for \textit{any} number of voters by splitting expected welfare into separate terms for each voter~\cite{boutilier2015optimal}. Here, the power $\frac1p$ in $W^\text{PM}$ prevents such a decomposition except when $p=1$ (Utilitarian; already known) and $p=0$ (Nash; after taking logarithms which gives an exact score vector dependent on $n$). We rely on the fact that, as $n$ grows, each voter's marginal contribution to $W^\text{PM}$ shrinks and $W^\text{PM}$ becomes near linear in per-voter terms $\psi_p(u_i(a))$ with negligible, bounded error (\Cref{lem:expansion}). As $\min$ and $\max$ functions are determined by a single element they admit no decomposition into per-voter terms, so ~\Cref{thm:optimal_is_psr,thm:optimal_deterministic} do not extend to $p=\pm\infty$. We cover these cases in \Cref{thm:egal_sw_direct,thm:maximax_sw_direct}.
\end{remark}

\subsection{Optimizing for Traditional SW Functions}
\label{sec:subsec_trad_funcs}

The result of \Cref{thm:optimal_deterministic} is quite abstract. We now discuss the implications of the theorem for each of the named social welfare functions discussed earlier. The first result, for utilitarian social welfare, is exactly that found in prior work from \citet{boutilier2015optimal}. In each case we continue to assume \Cref{a3} holds.

\begin{cor}[Utilitarian; $p=1$]
\label{cor:util}
For every $n$, the positional scoring rule with vector $\mathbf{s}^*(1) = (\mathbb{E}[U_{(1)}], \mathbb{E}[U_{(2)}], \dots, \mathbb{E}[U_{(m)}])$ maximizes $\mathrm{ESW}_1$ over all social choice functions.
\end{cor}


\begin{cor}[Nash; $p=0$]
\label{cor:nash}
The positional scoring rule with score vector $\mathbf{s}^*(0) = (\mathbb{E}[\log U_{(1)}], \mathbb{E}[\log U_{(2)}], \dots, \mathbb{E}[\log U_{(m)}])$ is asymptotically average-case optimal for $\mathrm{ESW}_0$ and is the only positional scoring rule with this property.
\end{cor}


Both \Cref{cor:util,cor:nash} follow from a direct application of \Cref{thm:optimal_deterministic}. We can also write a version of the Nash-optimal rule which is exactly optimal for a fixed, finite number of voters.

\begin{cor}[Exact Nash; $p=0$]
\label{cor:nash_exact}
For each $n$, the positional scoring rule with score vector
\[
s^{(n)}_k = n \log \mathbb{E}\big[U_{(k)}^{1/n}\big], \qquad k = 1, \dots, m,
\]
maximizes $\mathrm{ESW}_0$ over all social choice functions on $n$ voters.
Additionally, as $n \to \infty$,  $s^{(n)}_k$ decreases to $s^*_k(0) = \mathbb{E}[\log U_{(k)}]$.
\end{cor}
 
\begin{proof}
Nash welfare is a product of one component per voter. Conditional on $\boldsymbol{\sigma}$, these components are independent and $u_i(a)$ is distributed as $U_{(\sigma_i(a))}$ (\Cref{lemma:p0}), so the expectation of the product is the product of expectations. Taking logs,

\[
\log \mathbb{E}\big[W^\text{PM}(\mathbf{u}, a, 0) \mid \boldsymbol{\sigma}\big]
= \sum_{i=1}^n \log \mathbb{E}\big[U_{(\sigma_i(a))}^{1/n}\big]
= \frac{1}{n} S_{\mathbf{s}^{(n)}}(a, \boldsymbol{\sigma}).
\]
\noindent
Thus the alternative with the highest expected Nash welfare is exactly the alternative with the highest score under $\mathbf{s}^{(n)}$, and $f_{\mathbf{s}^{(n)}}$ is optimal at every profile.
\end{proof} 

Like utilitarian welfare, Nash welfare also admits an optimal score vector for any $n$. Unlike the utilitarian vector the Nash-optimal vector depends on $n$.
The results for the case of $p \to \pm \infty$ are more involved.
For \Cref{cor:maximax,cor:egal} we require \Cref{a1plus,a2}; that is, each voter samples utilities i.i.d. from $D$ with a continuous CDF $F$.

\begin{cor}[Maximax; $p \to +\infty$]
\label{cor:maximax}
As $p \to +\infty$, the score vector $\mathbf{s}^*(p)$ of the unique asymptotically average-case optimal positional scoring rule, normalized so that $s^*_1(p) = 1$ and $s^*_m(p) = 0$, converges to the plurality vector $\mathbf{s}^\text{p} = (1, 0, \dots, 0)$.
\end{cor}

\begin{proof}
Since $s_k^*(p)=(\mathbb{E}[U_{(k)}^p]-1)/p$, the normalized scores are $(r_k(p) - r_m(p))/(1 - r_m(p))$ where
$r_k(p):=\mathbb{E}[U_{(k)}^p]/\mathbb{E}[U_{(1)}^p]$, so it suffices to show $r_k(p)\to 0$ for $k>1$.
Using the standard order-statistic density of $U_{(k)}$ under our descending convention,

\begin{align*}
  \mathbb{E}[U_{(k)}^p] &= \frac{m!}{(k-1)!\,(m-k)!} \int_0^\infty x^p\, F(x)^{m-k}(1-F(x))^{k-1}\, dF(x), \\
  \mathbb{E}[U_{(1)}^p] &= m\int_0^\infty x^p\, F(x)^{m-1}\, dF(x).
\end{align*}

Both integrals share the factor $x^p\, dF(x)$.
Define the probability measure $dQ_p(x) = x^p\, dF(x)\,/\,\mathbb{E}[U^p]$,
which is well-defined since $\mathbb{E}[U^p]<\infty$ and sample $X \sim Q_p$.
Factoring $\mathbb{E}[U^p]$ from both integrals and cancelling, the ratio becomes
\[
  r_k(p)
  \;=\;
  \tbinom{m-1}{k-1}
  \frac{
    \mathbb{E}_{Q_p}\!\left[F(X)^{m-k}(1-F(X))^{k-1}\right]
  }{
    \mathbb{E}_{Q_p}\!\left[F(X)^{m-1}\right]
  }.
\]
As $p\to+\infty$, $Q_p$ concentrates near $\sup\operatorname{supp}(F)$, the largest possible utility, where $F(x)\to 1$,
so $(1-F(x))^{k-1}\to 0$ while the denominator $\to 1$.
So $r_k(p)\to 0$ for $k>1$, giving $\mathbf{s}^\text{p}$ in the limit.
\end{proof}

\begin{cor}[Egalitarian; $p\to-\infty$]
\label{cor:egal}
As $p \to -\infty$, the score vector $\mathbf{s}^*(p)$ of the unique asymptotically average-case optimal positional scoring rule, normalized so that $s^*_1(p) = 1$ and $s^*_m(p) = 0$, converges to the veto vector $\mathbf{s}^\text{v} = (1, \dots, 1, 0)$.
\end{cor}

 The proof is similar to the proof of \Cref{cor:maximax}.
\Cref{cor:maximax,cor:egal} describe the optimal score vector $\mathbf{s}^*(p)$ as $p \to \pm\infty$, they are not identical to expected maximax or egalitarian welfare as our results take $n \to \infty$ for a fixed and finite $p$. The corollaries also require \Cref{a3} which mildly restricts $D$. For example, \Cref{cor:egal} fails on the uniform distribution on $(0,1)$ but holds whenever the support of $D$ lies in $[\varepsilon, M]$ for $0 < \varepsilon < M < \infty$. \Cref{app:theorems} shows that veto and plurality are asymptotically optimal for exact egalitarian/maximax cases. These results complement \citet{apesteguia2011justice} who describe optimal non-scoring rules for egalitarian and maximax settings.

\begin{figure*}[tbp]
    \centering
    \includegraphics[width=.9\linewidth]{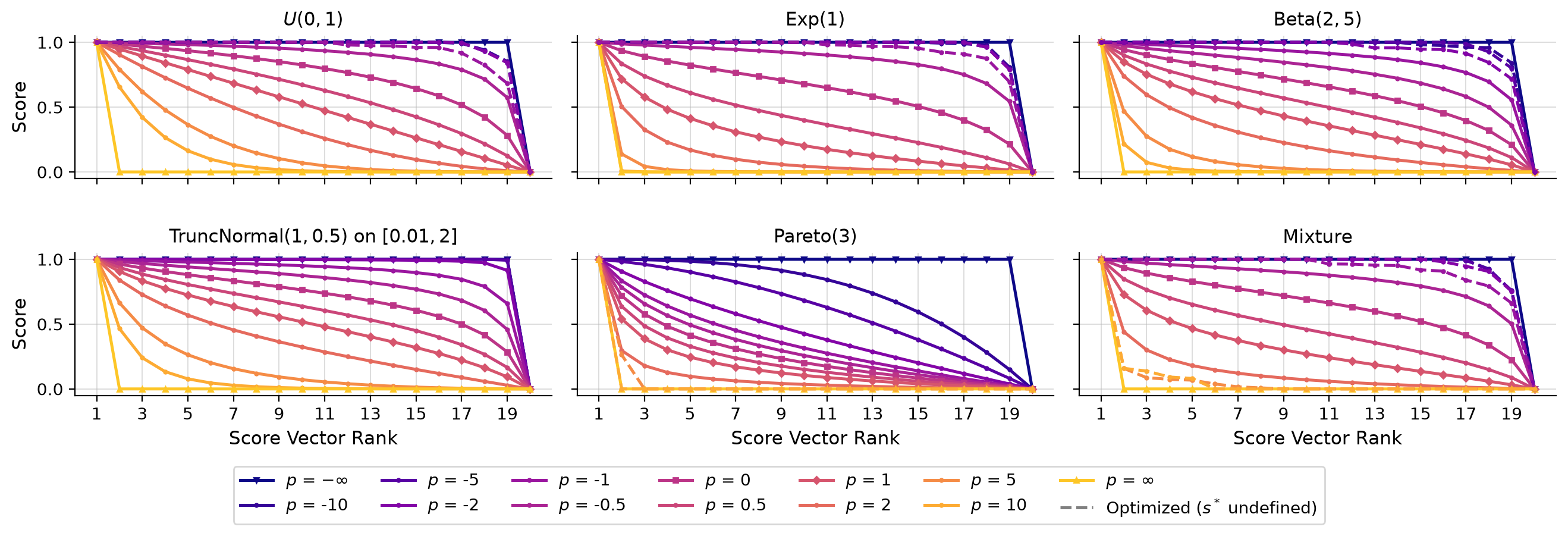}
    \caption{Shape of $\mathbf{s}^*(p)$ across varying values of $p$. The optimal score vector holds a relatively constant shape across a range of utility distributions. Dashed lines indicate experimentally optimal results where $s^*(p)$ is not well-defined on the distribution.
    }
    \label{fig:optimal_vector_shape}
\end{figure*}

\begin{figure}[tbp]
    \centering
    \includegraphics[width=\linewidth]{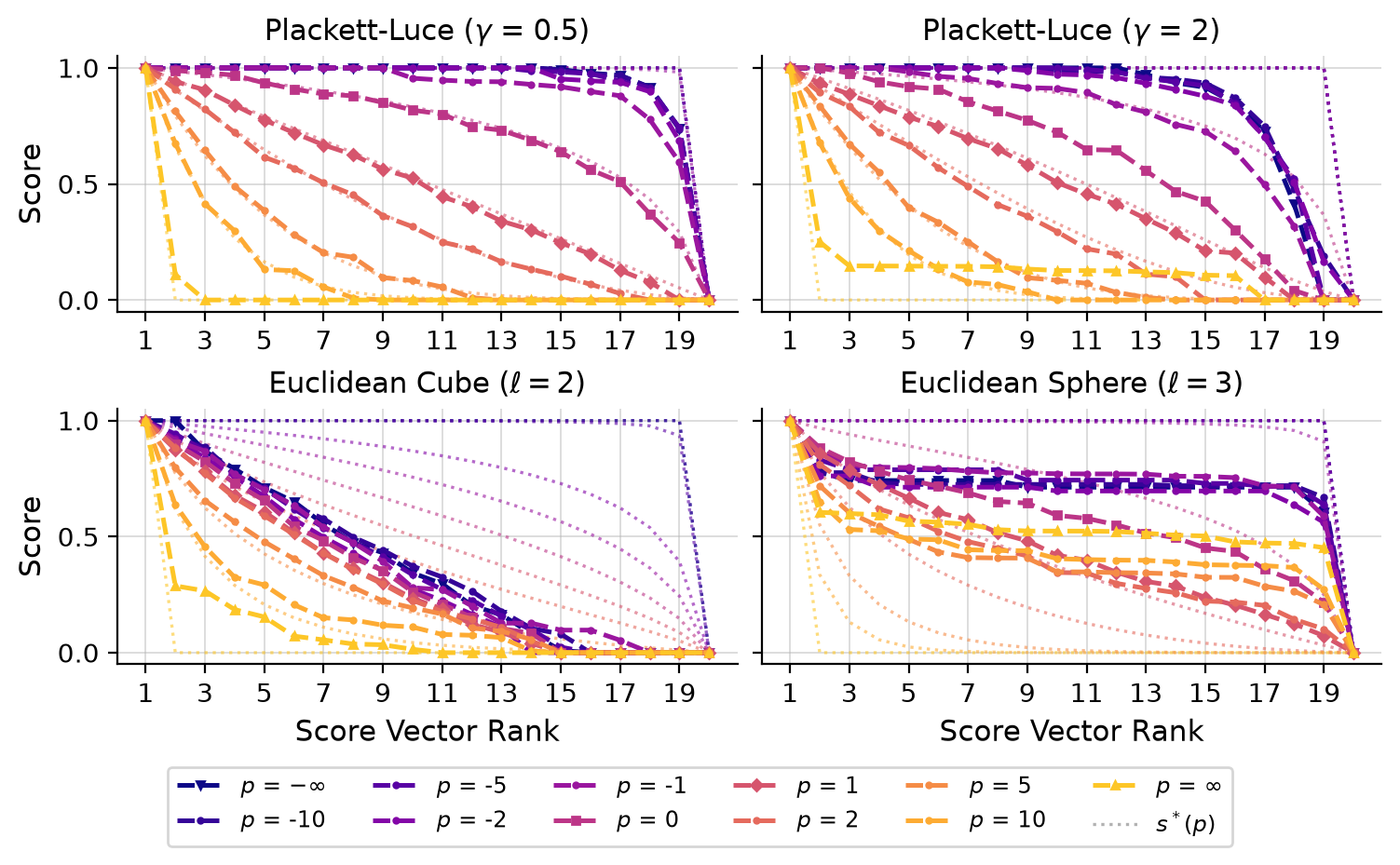}
    \caption{
    Score vectors optimized for non-i.i.d. utility sampling. The optimal vector on Plackett-Luce utilities begins to diverge from $s^*(p)$ as $\gamma$ increases while on Euclidean utilities it significantly diverges from $\mathbf{s}^*(p)$ in most cases.
    }
    \label{fig:optimal_nonneutral}
\end{figure}

\section{Optimizing Randomized Rules}

We now provide a characterization of the optimal randomized PSR. Under mild assumptions on utility and anonymity, $k$-approval maximizes ESW among randomized PSRs. Assuming voter independence and neutrality, the randomized plurality rule maximizes ESW. 

While \Cref{thm:optimal_deterministic} shows that deterministic PSRs maximize expected welfare among all social choice functions, randomized PSRs are nonetheless of interest as they have desirable normative properties. For instance, randomized PSRs are strategy-proof while deterministic PSRs are not~\cite{barbera1979majority,gibbard1977manipulation}.
We first introduce notation to aid our theorems.
Let $\mathbf{s}^{(k)} = (1, \dots, 1, 0, \dots, 0)$ denote the $k$-approval vector, with $k$ ones followed by $m - k$ zeros, so that $\mathbf{s}^{(1)}$ is plurality.
For an arbitrary voter $i$ and position $j$, let $b_{ij}$ be the alternative that voter $i$ ranks in position $j$. Fix $p \in [-\infty, \infty]$ and define
\[
v_j := \mathbb{E}\big[W^\text{PM}(\mathbf{u}, b_{ij}, p)\big],
\qquad
\bar v_k := \frac{1}{k} \sum_{j=1}^{k} v_j ,
\]
\noindent
the expected welfare of a voter's $j$-th choice, and the average expected welfare over a voter's top $k$ choices. 
This gives an alternative equation for expected social welfare with anonymous voters:
\[
\mathrm{ESW}_p(f^\mathrm{R}_\mathbf{s}) = 
\sum_{j=1}^m \frac{s_j}{\bar s} \cdot \frac{1}{n} \sum_{i=1}^n \mathbb{E}\big[W^\text{PM}(\mathbf{u}, b_{ij}, p)\big] = 
\frac{\sum_{j=1}^m s_j v_j}{\sum_{j=1}^m s_j}
\]

 
\begin{theorem}
\label{thm:optimal_randomized}
Fix $p \in [-\infty, \infty]$. Assume utilities are anonymous, rankings are strict, and $\mathbb{E}[u_i(a)] < \infty$ for all $i \in V$ and $a \in A$.
\begin{enumerate}
    \item[(a)] \emph{(No weight on suboptimal positions.)} Let $v^* = \max_j v_j$. For every $\mathbf{s}$ with $\sum_j s_j > 0$, $\mathrm{ESW}_p(f^\mathrm{R}_\mathbf{s}) \le v^*$, with equality if, and only if, $s_j = 0$ for every $j$ with $v_j < v^*$.
    \item[(b)] \emph{(k-approval maximizes ESW.)} Let $\bar v^* = \max_k \bar v_k$ and $K^* = \{k : \bar v_k = \bar v^*\}$. For every feasible score vector $\mathbf{s}$, $\mathrm{ESW}_p(f^\mathrm{R}_\mathbf{s}) \le \bar v^*$. Randomized $k$-approval attains this if, and only if, $k \in K^*$.
\end{enumerate}

\end{theorem}
 
\begin{proof}
 
\textbf{(a)} We have $\mathrm{ESW}_p(f^\mathrm{R}_\mathbf{s}) = \sum_j w_j v_j$ with weights $w_j = s_j / \bar s$, which are nonnegative and sum to $1$. Therefore
\[
v^* - \mathrm{ESW}_p(f^\mathrm{R}_\mathbf{s}) = \sum_{j=1}^m w_j \,(v^* - v_j).
\]
Each product $w_j (v^* - v_j)$ is nonnegative, so the sum is zero, i.e., $\mathrm{ESW}_p(f^\mathrm{R}_\mathbf{s}) = v^*$, if and only if every product is zero; that is, if and only if $w_j = 0$ whenever $v_j < v^*$.
 
\textbf{(b)} Let $\mathbf{s}$ be a feasible score vector, and write $d_k$ for the drop in score after position $k$: $d_k = s_k - s_{k+1}$ for $k < m$, and $d_m = s_m$. The drops from position $j$ onwards add up to the score at position $j$, so
\begin{equation}
\label{eq:drops}
s_j = \sum_{k = j}^m d_k
\qquad\text{and}\qquad
\mathbf{s} = \sum_{k=1}^m d_k\, \mathbf{s}^{(k)} .
\end{equation}
With $m = 3$, $(1, \tfrac12, 0)$ has $d = (\tfrac12, \tfrac12, 0)$ and is $\tfrac12 (1,0,0) + \tfrac12 (1,1,0)$.
 
We reformulate ESW in terms of drops. Substituting \eqref{eq:drops} and swapping the order of summation (both run over all pairs $j \le k$),
\[
\sum_{j=1}^m s_j v_j = \sum_{j=1}^m \sum_{k=j}^m d_k v_j = \sum_{k=1}^m d_k \sum_{j=1}^k v_j = \sum_{k=1}^m d_k\, k\, \bar v_k ,
\]
and, in the same way, $\sum_j s_j = \sum_k d_k\, k$. Hence
\[
\mathrm{ESW}_p(f^\mathrm{R}_\mathbf{s}) = \sum_{k=1}^m \omega_k\, \bar v_k ,
\qquad
\omega_k := \frac{d_k\, k}{\sum_{k'=1}^m d_{k'}\, k'} .
\]
The weights $\omega_k$ are nonnegative and sum to $1$ (the denominator is $\sum_j s_j > 0$). In the example above, $\omega = (\tfrac13, \tfrac23, 0)$, so the welfare of $(1, \tfrac12, 0)$ is $\tfrac13 \bar v_1 + \tfrac23 \bar v_2$.
Now (a) applies with $\bar v_k$ in place of $v_j$:
\[
\bar v^* - \mathrm{ESW}_p(f^\mathrm{R}_\mathbf{s}) = \sum_{k=1}^m \omega_k\, (\bar v^* - \bar v_k) \geq 0,
\]
with equality if, and only if, $\omega_k = 0$ (equivalently $d_k = 0$) for every $k \notin K^*$. The $k$-approval vector $\mathbf{s}^{(k)}$ has one drop, $d_k = 1$, so $\omega_k = 1$. Its welfare is $\bar v_k$, which equals $\bar v^*$ if, and only if, $k \in K^*$.
\end{proof}

When voters are anonymous, \Cref{thm:optimal_randomized} shows the randomized PSR which maximizes ESW is $k$-approval. Our assumptions allow distributions where voter utilities are not independent. Intuitively, $k > 1$ may result in optimal welfare when alternatives ranked lower by a voter are better for a group than their favourite. We explore this in \Cref{sec:experiments}. When voter utilities are independent and neutral, there is never an incentive to use $k > 1$ for randomized PSRs.

 
\begin{cor}
\label{cor:plurality_randomized}
Fix $p \in [-\infty, \infty]$. Assume \Cref{a1}, \Cref{a2}, $\mathbb{E}[U] < \infty$. For all $n$, the plurality vector uniquely maximizes $\mathrm{ESW}_p(f^\mathrm{R}_\mathbf{s})$ over all score vectors.
i.e., the optimal randomized PSR is random dictatorship.
\end{cor}
 
\begin{proof}
Under these assumptions \Cref{thm:optimal_randomized} applies. Consider $p \in \mathbb{R}$ first. We show $v_1 > v_j$ for every $j > 1$. Voter $1$ gives $b_{1j}$ utility $U_{(j)}$, the $j$-th largest of their $m$ utilities, while the other voters' utilities for $b_{1j}$ are independent of voter $1$ and, by neutrality, have the same distribution for every $j$. Using the same draws for both, $U_{(1)} > U_{(j)}$ by \Cref{a2}, and $W^\text{PM}$ strictly increases when one voter's utility increases, so the welfare of $b_{11}$ is always higher than $b_{1j}$.
Taking expectations gives $v_1 > v_j$.
When $p = \pm\infty$: $\min$ and $\max$ never decrease when one voter's utility increases, and with positive probability voter~$1$'s utility for $b_{1j}$ is the minimum ($p = -\infty$), or $b_{11}$ is the maximum ($p = +\infty$), again $v_1 > v_j$.
 Thus $v^* = v_1$ is attained only at $j = 1$, by \Cref{thm:optimal_randomized} the optimal score vectors are exactly those with $s_j = 0$ for all $j > 1$, i.e., positive multiples of $\mathbf{s}^\text{p}$.
\end{proof}

\section{Experiments}
\label{sec:experiments}

Our theoretical results describe the welfare-maximizing score vector across a wide range of settings but provide little intuition around what these optimal score vectors \textit{actually are} and how much benefit there is to using a score vector optimized for a particular setting over one of the existing canonical vectors shown in~\Cref{fig:canonical_vectors}. While applicable to many settings, our results are also limited by our assumptions around voter preferences which we can relax using experimental approaches. Here we explore our optimization setting experimentally.
All experiments are run with $n=100$ and $m=20$. We include additional discussion and results in \Cref{app:experiments}.

\subsection{Optimal Vectors Across $p$}

We prove in \Cref{sec:subsec_trad_funcs} the asymptotically average-case optimal score vector across the power mean family is $s_k^*(p) = \mathbb{E}[\psi_p(U_{(k)})]$. Here we visualize $\mathbf{s}^*$ across several utility distributions. In each, utilities are sampled i.i.d. for all voters. The mixture distribution selects another distribution uniformly at random for each voter.

\setlength{\columnsep}{-18pt}
\begin{multicols}{2}
\begin{itemize}[leftmargin=1em]
    \item $Uniform(0, 1)$
    \item Exponential with $\lambda = 1$
    \item $Beta(2, 5)$
    \item Trunc. Normal(1, 0.5) on $(0.01, 2)$
    \item Pareto with $\alpha=3$ on $[1, \infty)$
    \item Uniform mixture distribution
\end{itemize}
\end{multicols}

We show the results in \Cref{fig:optimal_vector_shape}. On values of $p$ where $\mathbf{s}^*$ is not defined we report the experimentally optimized score vector using the procedure described in \Cref{sec:relax_ass} with a dashed line. Hard-coded vectors for $p = \pm \infty$ are included for illustrative purposes.
The optimal vector takes a similar shape across distributions, suggesting a similar vector may perform well across settings.
For low $p$ score values remain large until late positions in the vector when it drops rapidly while the opposite holds true for higher values of $p$ ($> 1$).

\begin{figure*}[tbp]
    \centering
    \includegraphics[width=.9\linewidth]{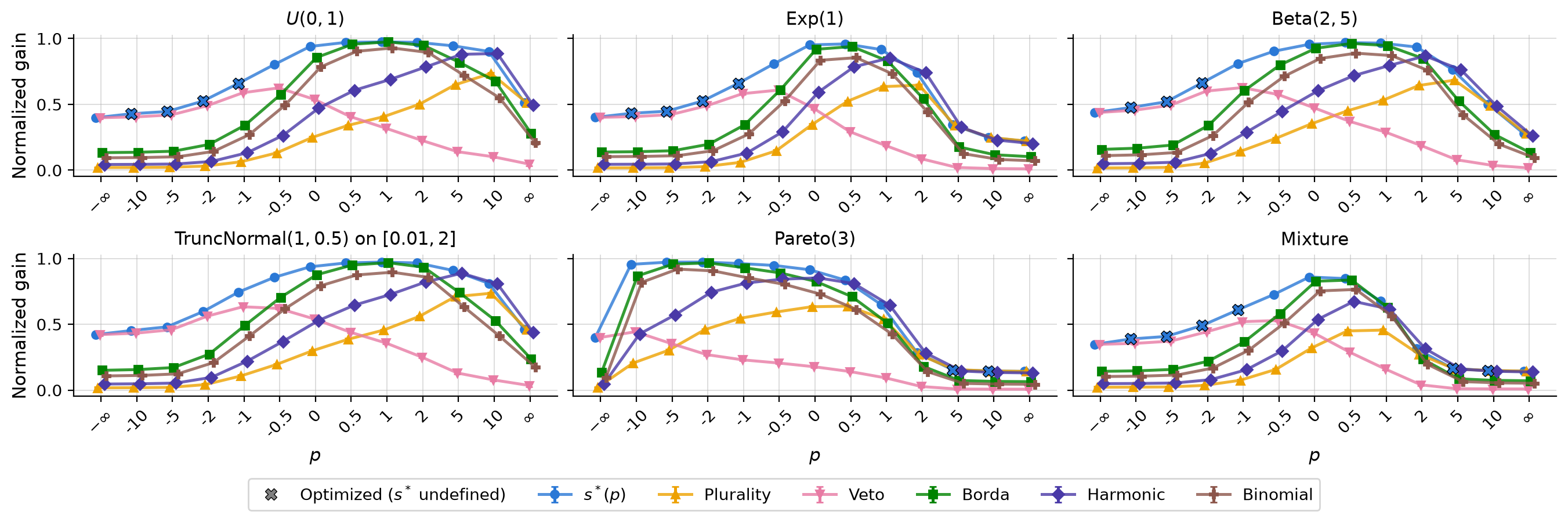}
    \caption{
    Normalized gain of each positional scoring rule on a scale between the welfare of a random winner and a perfect oracle. 
    Each canonical rule is significantly worse than the theoretically optimal $\mathbf{s}^*$ for at least some value of $p$. When $\mathbf{s}^*$ is marked with $\boldsymbol{\times}$ it is not well defined and we report welfare from an experimentally optimized score vector.
    }
    \label{fig:welfare_loss}
\end{figure*}

\begin{figure}[tbp]
    \centering
    \includegraphics[width=\linewidth]{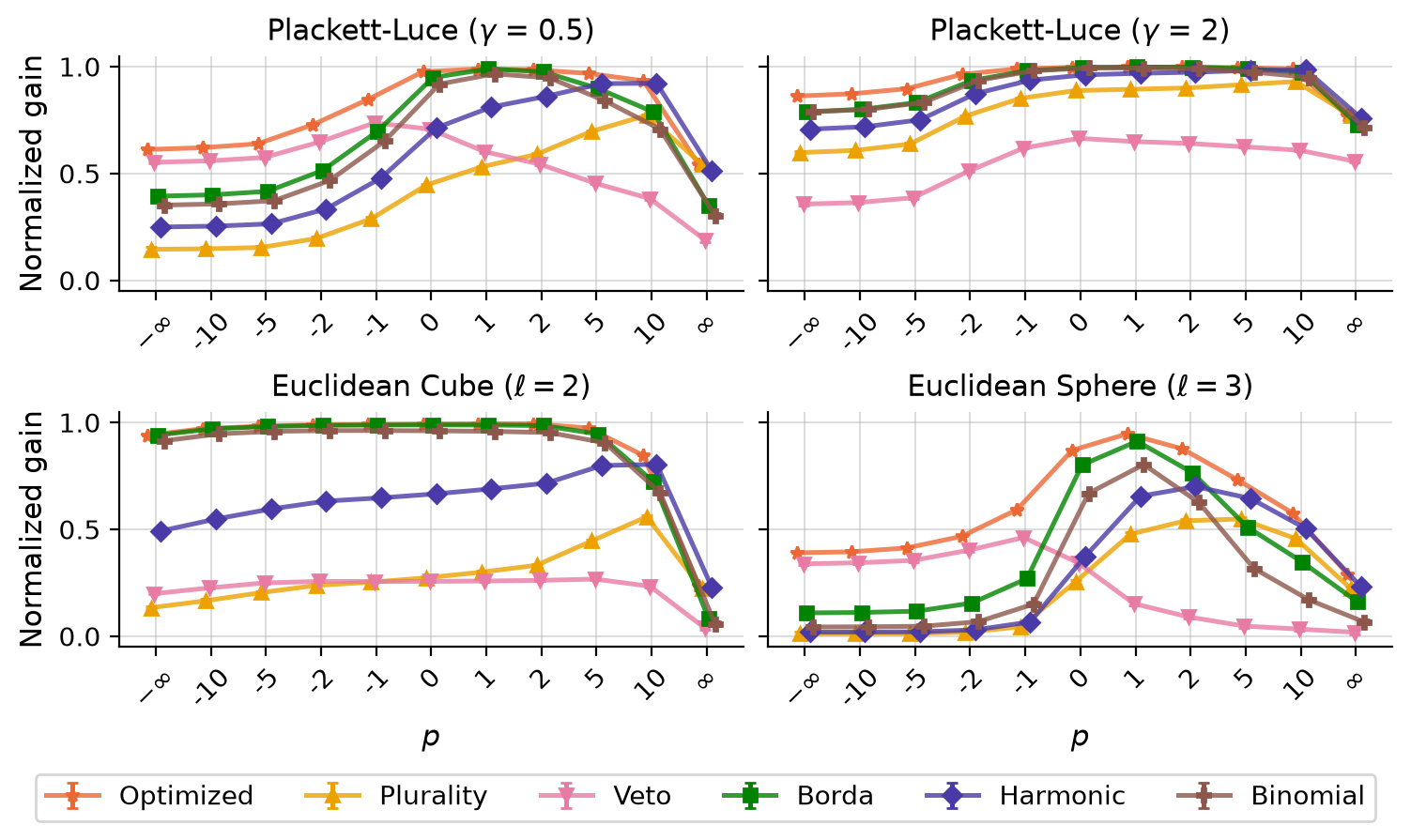}
    \caption{
    Normalized gain of each positional scoring rule on non-i.i.d. utility distributions. $\mathbf{s}^*$ is not known to be optimal on these distributions and is included among a comparison with the experimentally optimized vector.
    }
    \label{fig:welfare_loss_nonneutral}
\end{figure}

\subsection{Relaxing Utility Assumptions}
\label{sec:relax_ass}

Experimental analysis allows us to drop our prior assumptions. Here we study the optimal vector for utility distributions which are not neutral (Plackett-Luce) and where voters are not independent (Euclidean).
We use experimental optimization methods to find highly optimized score vectors and explore how they change while $p$ changes. We compare this with $\mathbf{s}^*$ to identify whether any relationship exists which could be used in future theoretical results.

In this experiment we sample utilities from two families of distributions:
\textbf{Plackett-Luce} utilities~\cite{luce2005individual,plackett1975analysis} with $\gamma \in \{0.5, 2\}$ with weight vector $\mathbf{w}$ is $w_j = \exp{\gamma (\frac12 - \frac{j-1}{m-1})}$ for $j \in 1 \dots m$. When $\gamma = 0$ this reduces to the uniform distribution. We also use \textbf{Euclidean} preferences where voters and alternatives are placed uniformly at random in an $\ell$-dimensional space and a voter at position $v_i$ receives utility $u_i = 1 - \frac{\Vert v_i - a_j \Vert_2}{d}$ for an alternative at position $a_j$ on a space with diameter $d$.
We do not have an exact method for computing $\mathbf{s}^*$ for these distributions. Instead we generate utilities for $200,000$ voters, with alternatives randomly regenerated for each voter under Euclidean preferences. We report $U_{(k)}$ as the mean of the $k$-th largest element across each list and use it to calculate $\mathbf{s}^*$.
 
We consider three basic methods for optimizing scoring rules: 

\begin{itemize}
    \item Simulated annealing, as discussed by \citet{armstrong2026optimizing}.
    \item Covariance Matrix Adaptation Evolution Strategy 
    \cite{hansen2001completely}
    \item Random search over vectors.
\end{itemize}

Methods are restricted to feasible rules and run for $10,000$ steps. The first two are run once starting from each canonical vector, a random vector, and $\mathbf{s}^*$. We record the score vector achieving highest welfare on a training set of $20,000$ profiles and report the welfare of that vector on a test set of $20,000$ profiles from the same distribution.
\Cref{app:experiments} shows that this optimization process closely matches theoretical results when utilities are neutral and independent.

In \Cref{fig:optimal_nonneutral} we see that the experimentally optimal score vector (dashed lines) can significantly diverge from $\mathbf{s}^*(p)$ (dotted lines). Plackett-Luce models have a similar shape, which begins to diverge as they become less neutral. Euclidean distributions present what appear to be truly novel dynamics. Here differences from $\mathbf{s}^*(p)$ are much larger than can be explained by simulation noise. We explore additional parameterizations of these families in \Cref{app:experiments}.

\begin{figure*}[tbp]
    \centering
    \includegraphics[width=.9\linewidth]{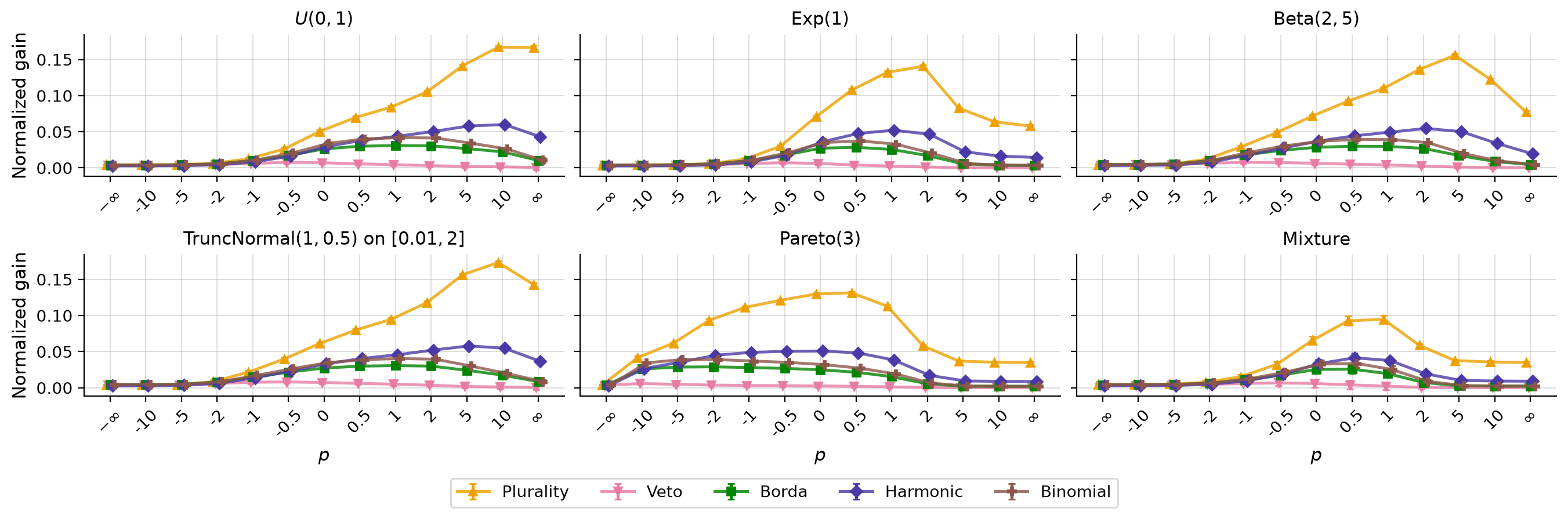}
    \caption{
    Normalized gain of randomized positional scoring rules on a scale between the welfare of a random winner and a perfect oracle. Plurality is optimal here and it significantly outperforms other rules on most values of $p$ across distributions.
    }
    \label{fig:welfare_loss_randomized_neutral}
\end{figure*}

\begin{figure}[tbp]
    \centering
    \includegraphics[width=\linewidth]{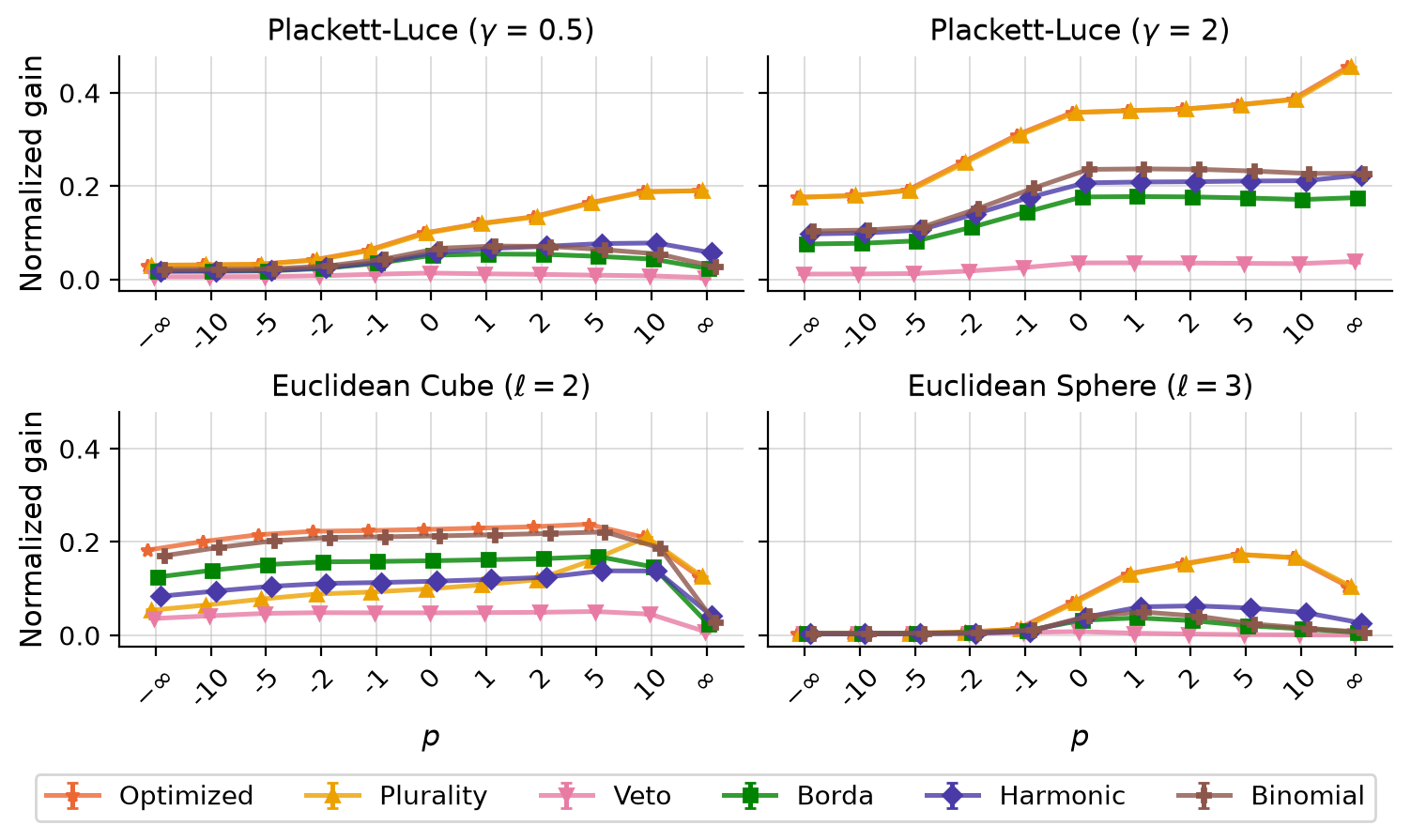}
    \caption{
    Normalized gain of each randomized positional scoring rule on non-i.i.d. utility distributions.
    \Cref{thm:optimal_randomized} shows that a $k$-approval vector optimizes welfare. In most cases $k=1$ (plurality) is best but our optimization finds $k>1$ can give higher welfare on the Euclidean Cube distribution.
    }
    \label{fig:welfare_loss_nonneutral_randomized}
\end{figure}

\subsection{How Important is Optimality?}

Thus far, we have illustrated the dynamics of the welfare-maximizing score vector across a range of power mean functions.
What we have not yet explored is the amount of welfare which is gained or lost by use of one score vector over another. If the difference in social welfare between any scoring vectors is insignificant, the optimal vector is quite unimportant. We show in this section that the optimal vector is typically significantly better than others.

Using the settings from the previous two experiments we place the welfare given by each score vector on a scale between a randomly selected winner and an oracle $f^*$ which selects the welfare-maximizing alternative on every profile. We call this measure normalized gain, with $G(\mathbf{s}) = \frac{ESW(f_\mathbf{s}) - ESW(f_\mathbf{1})}{ESW(f^*) - ESW(f_\mathbf{1})}$. A value of zero indicates that $\mathbf{s}$ provides no benefit over a random selection while a value of $1$ corresponds to electing the welfare-maximizing alternative on every profile. $G(\mathbf{s})$ links naturally with our definition of average-case optimality (\Cref{def:opt}) and has been used in prior social choice experiments by \citet{weber1978comparison} and \citet{merrill1984comparison}.

We see in \Cref{fig:welfare_loss} (i.i.d.) and \Cref{fig:welfare_loss_nonneutral} (non-i.i.d.) a large gap between canonical score vectors and the optimal score vector in most settings. Error bars show two standard errors of each rule's mean welfare over $20,000$ test profiles. No single canonical rule performs near optimal in all settings. One special case of interest is the Borda vector on the Euclidean Cube distribution: Here Borda performs extremely similarly to the optimized vector in all case, suggesting that Borda ---or a Borda-like vector--- may be exactly optimal in this setting. This result mirrors the experimental findings of \citet{green2016statistical} in exploring the utilitarian social welfare of Borda in a single-dimensional setting; \Cref{app:experiments} shows that this holds across higher dimensions of the Cube distribution.

\subsection{Randomized Rules}

The previous experiments study optimal deterministic rules. We have shown that, among randomized PSRs, a $k$-approval rule maximizes ESW across the power mean family. Thus, there is less room to explore how the optimal rules changes across the power mean. However, the question of whether the best randomized PSR gives significantly more welfare than other canonical score vectors is yet unexplored.
In fact, the optimal score vector does consistently provide a significant improvement to welfare over other score vectors on at least some values of $p$. On neutral distributions with independent voters, \Cref{fig:welfare_loss_randomized_neutral} shows that plurality is often much better in terms of welfare than other rules. When utility assumptions are relaxed, \Cref{fig:welfare_loss_nonneutral_randomized} shows a similar outcome. Plurality, or another $k$-approval vector is usually more performant than other score vectors on some cases. In Euclidean preferences where each voter's most favoured alternative may not be the winner which maximizes social welfare we typically see that the best rule involves $k > 1$. Randomized rules notably provide much lower normalized gain than deterministic rules across all distributions.

\section{Discussion}

In this paper we have analyzed from the perspective of ordinal voting a highly general form of social welfare function: the power mean. We show that, with an asymptotic number of voters, the social choice function which maximizes expected social welfare is always a deterministic PSR and we characterize the gap between optimal social welfare provided by an oracle and the welfare of the best ordinal social choice function. The optimal score vector is a direct transform of the order statistics of voter utility distributions. We present a novel result describing exact and asymptotic scoring functions for Nash social welfare and strengthen results showing that egalitarian and maximax welfare are optimized by veto and plurality.
We also show that among randomized PSRs, a $k$-approval rule will always maximize expected social welfare. When voter utilities are independent and neutral the plurality rule is optimal.

Experimentally, we find that optimal scoring vectors for deterministic rules often take similar shapes across utility distributions. This shape holds even on some distributions (Plackett-Luce) which are not covered by our assumptions on utility. On utility distributions in which voters and alternatives are highly structured (Euclidean preferences) we find that novel shapes of score vector, not described by our theoretical findings, maximize social welfare.

Our results highlight some novel directions worth attention:

\textit{Extension to Additional Distributions.}
Our assumptions on voter utility cover a wide range of distributions but many results rely on independence between voters, an assumption quite misaligned with reality. Meanwhile, our experimental results show that novel score vectors emerge when voter utilities are not independent. The well-defined domain of Euclidean preferences may be amenable to theoretical analysis on the question of which social choice function maximizes welfare and our experimental answers to this question suggest interesting and novel outcomes.

\textit{Bounded Randomness.}
Our work shows that deterministic score vectors outperform randomized ones in both a theoretical and empirical sense. However, randomized rules provide satisfying axiomatic properties around strategy-proofness and fairness. Developing and analyzing a novel blend of deterministic and randomized approaches -- e.g., a temperature which controls the ``level'' of randomness --  may be able to provide a high level of expected or worst-case welfare while also maintaining axiomatic properties not found in deterministic scoring rules.


\clearpage
\bibliographystyle{ACM-Reference-Format} 
\bibliography{refs}

\clearpage

\appendix
\crefalias{section}{appendix}

\section{Missing Proofs of Theoretical Claims}
\label{app:theorems}

This appendix contains all proofs which are missing from the main body. 
The properties of the Box-Cox utility transform are shown in \Cref{app:box_cox}. Young's lemma and our use of the lemma to demonstrate that a positional scoring rule maximizes expected social welfare appear in \Cref{app:subsec_optimal_is_psr}. We prove the shape of the optimal scoring rule in \Cref{app:subsec_optimal_rule_shape}.

\paragraph{Asymptotic Optimality.}
Each asymptotic optimality result in this appendix uses \Cref{def:opt} in the same way. Fix $p \in \mathbb{R} \cup \{\pm\infty\}$, a reference rule $f^\circ$ and a rate $c_n \to \infty$ with $c_n\big(\mathrm{OPT}_n - \mathrm{ESW}_p(f^\circ)\big) \to 0$. For any rule $f$, let
\[
\Delta(f) = \lim_{n \to \infty} c_n\big(\mathrm{ESW}_p(f^\circ) - \mathrm{ESW}_p(f)\big)
\]
whenever this limit exists. Then $c_n\big(\mathrm{OPT}_n - \mathrm{ESW}_p(f)\big) \to \Delta(f)$ as well. So if $\Delta(f_\mathbf{1}) > 0$, multiplying the numerator and denominator in \Cref{def:opt} by $c_n$ gives
\[
\lim_{n \to \infty} \frac{\mathrm{OPT}_n - \mathrm{ESW}_p(f)}{\mathrm{OPT}_n - \mathrm{ESW}_p(f_\mathbf{1})} = \frac{\Delta(f)}{\Delta(f_\mathbf{1})},
\]
and $f$ is asymptotically average-case optimal if, and only if, $\Delta(f) = 0$. For finite $p$ we use $c_n = n^{\frac{1}{2}}$ and $f^\circ = f_{\mathbf{s}^*(p)}$ (\Cref{thm:optimal_deterministic_appendix}). For $p = -\infty$ and $p = +\infty$ we use $c_n = n^{\frac{3}{2}}$, with $f^\circ$ being veto and plurality respectively (\Cref{thm:egal_sw_direct,thm:maximax_sw_direct}).

\subsection{Utility Transformation}
\label{app:box_cox}

The first lemma below provides a simple utility transformation used throughout the rest of the proofs. While we use this function, other functions with identical properties could also provide similar results.

\begin{lemma}
\label{lem:hp_appendix}
Fix $p\in\mathbb{R}$ and let $\psi_p:(0,\infty)\to\mathbb{R}$ be the Box-Cox transform of \Cref{eq:psi}.

Then
\begin{enumerate}
\item[(a)] $\psi_p$ is strictly increasing on $(0,\infty)$; and

\item[(b)] for every utility profile $\mathbf u$ and every $a\in A$,

\[
\psi_p\big(W^{\mathrm{PM}}(\mathbf u,a,p)\big)\;=\;\frac{1}{n}\sum_{i=1}^{n}\psi_p\big(u_i(a)\big).
\]

\end{enumerate}

Consequently, for all $a,b\in A$,

\begin{multline*}
W^{\mathrm{PM}}(\mathbf u,a,p)\ \ge\ W^{\mathrm{PM}}(\mathbf u,b,p)
\\ \Longleftrightarrow\quad
\sum_{i=1}^{n}\psi_p\big(u_i(a)\big)\ \ge\ \sum_{i=1}^{n}\psi_p\big(u_i(b)\big).
\end{multline*}
\end{lemma}

\begin{proof}
(a) For $p\neq 0$, $\psi_p$ is differentiable on $(0,\infty)$ with $\psi_p'(t)=t^{p-1}>0$. For $p=0$, $\log$ is
strictly increasing.

(b) All utilities are positive, so $W^{\mathrm{PM}}(\mathbf u,a,p)>0$ and $\psi_p$ may be applied to it. For $p\neq 0$,

\begin{multline*}
\psi_p\big(W^{\mathrm{PM}}(\mathbf u,a,p)\big) =\frac{1}{p}\Big(\Big(\big(\tfrac{1}{n}\textstyle\sum_i u_i(a)^{p}\big)^{1/p}\Big)^{p} - 1\Big)
\\ =\frac{1}{n}\sum_i \frac{u_i(a)^{p} - 1}{p}
=\frac{1}{n}\sum_i \psi_p\big(u_i(a)\big).
\end{multline*}

For $p=0$, $\psi_0\big((\prod_i u_i(a))^{1/n}\big)=\frac{1}{n}\log\prod_i u_i(a)=\frac{1}{n}\sum_i \psi_0\big(u_i(a)\big)$.

The final equivalence follows by applying the strictly increasing function $\psi_p$ to both sides of the welfare comparison, rewriting each side using (b), and multiplying both sides by $n$.
\end{proof}

\subsection{The Optimal Rule is a PSR}
\label{app:subsec_optimal_is_psr}

\citet{young1975social} provided an axiomatic characterization of scoring rules which we use to prove that the rule maximizing expected social welfare is a scoring rule.

A \emph{social choice correspondence} maps each preference profile, over any number of voters, to a nonempty set of alternatives. It is a \emph{scoring correspondence} if there is a score vector $\mathbf{s}$, the same for every number of voters, such that it always returns exactly the alternatives with maximal score $S_\mathbf{s}(a, \boldsymbol{\sigma})$.

\begin{lemma}[Young]
\label{lem:young_appendix}
For preference profiles $\boldsymbol{\sigma}_1$ and $\boldsymbol{\sigma}_2$ write $\boldsymbol{\sigma}_1 + \boldsymbol{\sigma}_2$ for the profile containing the voters of both, $t\boldsymbol{\sigma}_1$ for the profile containing $t \in \mathbb{N}$ copies of each voter of $\boldsymbol{\sigma}_1$, and, for a permutation $\pi$ of $A$, $\pi\boldsymbol{\sigma}_1$ for the profile obtained by replacing every alternative $b$ by $\pi(b)$ in every ranking. Let $g$ be a social choice correspondence defined on profiles of every size. Call $g$

\begin{itemize}
\item \emph{anonymous} if $g(\boldsymbol{\sigma}_1)$ is unchanged when the voters of $\boldsymbol{\sigma}_1$ are reordered;

\item \emph{neutral} if $g(\pi\boldsymbol{\sigma}_1)=\pi\big(g(\boldsymbol{\sigma}_1)\big)$ for every permutation $\pi$ of $A$ and every $\boldsymbol{\sigma}_1$;

\item \emph{consistent} if $g(\boldsymbol{\sigma}_1+\boldsymbol{\sigma}_2)=g(\boldsymbol{\sigma}_1)\cap g(\boldsymbol{\sigma}_2)$ whenever $g(\boldsymbol{\sigma}_1)\cap g(\boldsymbol{\sigma}_2)\neq\emptyset$;

\item \emph{continuous} if $g(\boldsymbol{\sigma}_1)=\{a\}$ implies that for every $\boldsymbol{\sigma}_2$ there is $T\in\mathbb{N}$ with $g(\boldsymbol{\sigma}_2+t\boldsymbol{\sigma}_1)=\{a\}$ for all $t\ge T$.
\end{itemize} 

Then $g$ is a scoring correspondence if and only if it is anonymous, neutral, consistent, and continuous.
\end{lemma}

\begin{theorem}\label{thm:optimal_is_psr_appendix}
Fix $p\in\mathbb{R}$ and suppose \Cref{a3} holds. There is a score vector $\mathbf{s} \in \mathbb{R}^m$ such that the positional scoring rule $f_\mathbf{s}$ is asymptotically average-case optimal (\Cref{def:opt}). When $p = 1$, $f_\mathbf{s}$ moreover maximizes $\mathrm{ESW}_1$ over all social choice functions for every $n$.
\end{theorem}

The proof applies \Cref{lem:young_appendix} to the objective $\mathbb{E}[\sum_{i=1}^n \psi_p(u_i(a)) \mid \boldsymbol{\sigma}]$, which is additive across voters. The following lemma shows that this objective determines $\mathrm{ESW}_p$ up to an additive term of $o(n^{-\frac{1}{2}})$.

\begin{lemma}
\label{lem:expansion}
Fix $p \in \mathbb{R}$ and suppose \Cref{a3} holds. Let
\[
V_p(\mathbf{u}, a) = \sum_{i=1}^n \psi_p(u_i(a)) = n\,\psi_p\big(W^{\mathrm{PM}}(\mathbf u,a,p)\big),
\]
where the second expression follows from \Cref{lem:hp_appendix}(b), and let $\mu_\psi = \mathbb{E}[\psi_p(U)]$, $C_p = \psi_p^{-1}(\mu_\psi)$ and $\kappa_p = (\psi_p^{-1})'(\mu_\psi) = C_p^{1-p} > 0$. There exist $\varepsilon_n \geq 0$ with $n^{\frac{1}{2}}\varepsilon_n \to 0$ such that for every $n$ and every (possibly randomized) social choice function $f$,

\begin{multline*}
\Big| \mathrm{ESW}_p(f) - C_p 
- \kappa_p \Big(\frac{1}{n}\mathbb{E}[V_p(\mathbf{u}, f(\boldsymbol{\sigma}))] - \mu_\psi\Big) \Big| \leq \varepsilon_n.
\end{multline*}
\end{lemma}

\begin{proof}
For each $a \in A$ let $\bar Y_n(a) = \frac{1}{n} V_p(\mathbf{u}, a)$, so that $W^\text{PM}(\mathbf{u}, a, p) = \psi_p^{-1}(\bar Y_n(a))$ by \Cref{lem:hp_appendix}(b). By \Cref{a1}, $\bar Y_n(a)$ is an average of $n$ independent copies of $\psi_p(U)$. Hence $\mathbb{E}[\bar Y_n(a)] = \mu_\psi$ and $\mathbb{E}[(\bar Y_n(a) - \mu_\psi)^2] = \mathrm{Var}(\psi_p(U))/n = O(n^{-1})$, where the variance is finite by \Cref{a3}.

\begin{enumerate}
    \item \textbf{Linearization.} $\psi_p^{-1}$ is twice continuously differentiable on the range of $\psi_p$, an open interval containing $\mu_\psi$. Fix $\delta > 0$ such that $[\mu_\psi - \delta, \mu_\psi + \delta]$ lies in this range. As $C_p = \psi_p^{-1}(\mu_\psi)$ and $\kappa_p = (\psi_p^{-1})'(\mu_\psi)$, Taylor's theorem gives, for every $a \in A$,
    \[
    W^\text{PM}(\mathbf{u}, a, p) = C_p + \kappa_p (\bar Y_n(a) - \mu_\psi) + \rho_p(\bar Y_n(a)),
    \]
    where $|\rho_p(y)| \leq c\,(y - \mu_\psi)^2$ whenever $|y - \mu_\psi| \leq \delta$, for a constant $c$ that does not depend on $n$. The identity $\kappa_p = C_p^{1-p}$ follows from $\psi_p'(t) = t^{p-1}$.

    \item \textbf{The remainder is negligible.} We show $\mathbb{E}[|\rho_p(\bar Y_n(a))|] = o(n^{-\frac{1}{2}})$. Let $B_n = \{|\bar Y_n(a) - \mu_\psi| > \delta\}$, so that $\Pr(B_n) = O(n^{-1})$ by Chebyshev's inequality. Outside $B_n$, the remainder contributes at most $c\,\mathbb{E}[(\bar Y_n(a) - \mu_\psi)^2] = O(n^{-1})$. On $B_n$, as $W^\text{PM} > 0$ and $C_p > 0$,
    \[
    |\rho_p(\bar Y_n(a))| \leq W^\text{PM}(\mathbf{u}, a, p) + C_p + \kappa_p |\bar Y_n(a) - \mu_\psi|.
    \]
    By Cauchy-Schwarz the last two terms contribute $O(n^{-1})$. For the first term, let $q$ be as in \Cref{a3}. Power means are nondecreasing in their parameter and $q > p$, so $\mathbb{E}[W^\text{PM}(\mathbf{u}, a, p)^q] \leq \mathbb{E}[W^\text{PM}(\mathbf{u}, a, q)^q] = \mathbb{E}[U^q] < \infty$. Hölder's inequality then gives
    \[
    \mathbb{E}[W^\text{PM}(\mathbf{u}, a, p)\, \mathbf{1}_{B_n}] \leq \mathbb{E}[U^q]^{\frac{1}{q}} \Pr(B_n)^{1 - \frac{1}{q}} = O(n^{\frac{1}{q} - 1}),
    \]
    which is $o(n^{-\frac{1}{2}})$ since $q > 2$.

    \item \textbf{Any social choice function.} Let $\omega = f(\boldsymbol{\sigma})$. Applying step 1 at $a = \omega$ and taking expectations, with $\mathbb{E}[\bar Y_n(\omega)] = \frac{1}{n}\mathbb{E}[V_p(\mathbf{u}, \omega)]$,
    \[
    \mathrm{ESW}_p(f) - C_p - \kappa_p\Big(\frac{1}{n}\mathbb{E}[V_p(\mathbf{u}, \omega)] - \mu_\psi\Big) = \mathbb{E}[\rho_p(\bar Y_n(\omega))].
    \]
    As $|\rho_p(\bar Y_n(\omega))| \leq \sum_{a \in A} |\rho_p(\bar Y_n(a))|$ and the $\bar Y_n(a)$ are identically distributed (\Cref{a1}), the claim holds with $\varepsilon_n = m\, \mathbb{E}[|\rho_p(\bar Y_n(a))|]$. This does not depend on $f$ and satisfies $n^{\frac{1}{2}}\varepsilon_n \to 0$ by step 2.
\end{enumerate}
\end{proof}

Now we are ready to prove \Cref{thm:optimal_is_psr}.

\begin{proof}[Proof of \Cref{thm:optimal_is_psr_appendix}]
For every preference profile $\boldsymbol{\sigma}$, of any size $n$, and every $a\in A$ let
\[
G(\boldsymbol{\sigma},a)\;:=\;\mathbb{E}\big[V_p(\mathbf u, a)\ \big|\ \boldsymbol{\sigma}\big]
\;=\;\sum_{i=1}^{n}\mathbb{E}\big[\psi_p\big(u_i(a)\big)\ \big|\ \boldsymbol{\sigma}\big],
\]
which is finite by \Cref{a3}. Let $f^{*}(\boldsymbol{\sigma})=\operatorname*{arg\,max}_{a\in A}G(\boldsymbol{\sigma},a)$ be the social choice correspondence returning \emph{all} maximizers. Following \citet{boutilier2015optimal}, we use \Cref{lem:young_appendix} to show that $f^*$ is a scoring correspondence. We then use \Cref{lem:expansion} to show that the corresponding positional scoring rule is asymptotically average-case optimal.

\noindent\textbf{$G$ is a sum of per-voter terms.}
Voters are independent and identically distributed (\Cref{a1}), so for voter $i$ conditioning on $\boldsymbol{\sigma}$ is the same as conditioning on $\sigma_i$, and
\begin{equation}
\label{eq:decomp}
G(\boldsymbol{\sigma},a)=\sum_{i=1}^n h_a(\sigma_i), \qquad h_a(\sigma) := \mathbb{E}\big[\psi_p(u_1(a))\ \big|\ \sigma_1=\sigma\big],
\end{equation}
where $h_a$ depends neither on $i$ nor on $n$. Therefore
\begin{equation}\label{eq:additive}
G(\boldsymbol{\sigma}_1+\boldsymbol{\sigma}_2,a)=G(\boldsymbol{\sigma}_1,a)+G(\boldsymbol{\sigma}_2,a) \qquad \forall a \in A.
\end{equation}

\noindent\textbf{$f^*$ is a scoring correspondence.} We check the conditions of \Cref{lem:young_appendix}.

\emph{Anonymous.} By \eqref{eq:decomp}, $G(\boldsymbol{\sigma}, a)$ does not depend on the order of the voters.

\emph{Neutral.} By \Cref{a1}, relabelling the alternatives by $\pi$ does not change the distribution of utilities, so $G(\pi\boldsymbol{\sigma}, \pi(a)) = G(\boldsymbol{\sigma}, a)$ and $f^*(\pi\boldsymbol{\sigma}) = \pi(f^*(\boldsymbol{\sigma}))$.

\emph{Consistent.} Suppose $f^{*}(\boldsymbol{\sigma}_1)\cap f^{*}(\boldsymbol{\sigma}_2)\neq\emptyset$. By \eqref{eq:additive}, $G(\boldsymbol{\sigma}_1+\boldsymbol{\sigma}_2,\cdot)$ is the sum of $G(\boldsymbol{\sigma}_1,\cdot)$ and $G(\boldsymbol{\sigma}_2,\cdot)$, which have a common maximizer. So $b$ maximizes the sum if, and only if, it maximizes both terms, i.e., $f^{*}(\boldsymbol{\sigma}_1+\boldsymbol{\sigma}_2)=f^{*}(\boldsymbol{\sigma}_1)\cap f^{*}(\boldsymbol{\sigma}_2)$.

\emph{Continuous.} Suppose $f^{*}(\boldsymbol{\sigma}_1)=\{a\}$, so that $\delta:=\min_{b\neq a}\big(G(\boldsymbol{\sigma}_1,a)-G(\boldsymbol{\sigma}_1,b)\big)>0$. By \eqref{eq:additive}, for every profile $\boldsymbol{\sigma}_2$ and every $b\neq a$,
\[
G(\boldsymbol{\sigma}_2+t\boldsymbol{\sigma}_1,a)-G(\boldsymbol{\sigma}_2+t\boldsymbol{\sigma}_1,b) \ \ge\ G(\boldsymbol{\sigma}_2,a)-G(\boldsymbol{\sigma}_2,b)+t\delta,
\]
which is positive for all sufficiently large $t$. So $f^{*}(\boldsymbol{\sigma}_2+t\boldsymbol{\sigma}_1)=\{a\}$ for all sufficiently large $t$.

By \Cref{lem:young_appendix} there is a score vector $\mathbf{s}$, the same for every $n$, with $f^*(\boldsymbol{\sigma}) = \operatorname*{arg\,max}_{a \in A} S_\mathbf{s}(a, \boldsymbol{\sigma})$. Breaking ties uniformly at random, $f_\mathbf{s}$ selects uniformly from $f^*(\boldsymbol{\sigma})$.

\noindent\textbf{$f_\mathbf{s}$ is asymptotically average-case optimal.}
Let $g$ be any (possibly randomized) social choice function on $n$ voters. Conditioning on $\boldsymbol{\sigma}$ and the randomness of $g$,
\[
\mathbb{E}[V_p(\mathbf u, g(\boldsymbol{\sigma}))] = \mathbb{E}\big[G(\boldsymbol{\sigma}, g(\boldsymbol{\sigma}))\big] \leq \mathbb{E}\big[\max_{a \in A} G(\boldsymbol{\sigma}, a)\big] = \mathbb{E}[V_p(\mathbf u, f_\mathbf{s}(\boldsymbol{\sigma}))].
\]
Applying \Cref{lem:expansion} to $g$ and to $f_\mathbf{s}$ gives $\mathrm{ESW}_p(g) - \mathrm{ESW}_p(f_\mathbf{s}) \leq 2\varepsilon_n$. Taking $g$ to be optimal, $0 \leq \mathrm{OPT}_n - \mathrm{ESW}_p(f_\mathbf{s}) \leq 2\varepsilon_n = o(n^{-\frac{1}{2}})$. The same argument, using \Cref{lemma:p0} in place of the display above, gives $\mathrm{OPT}_n - \mathrm{ESW}_p(f_{\mathbf{s}^*(p)}) = o(n^{-\frac{1}{2}})$.

We now apply the argument at the start of \Cref{app:theorems} with $c_n = n^{\frac{1}{2}}$ and $f^\circ = f_{\mathbf{s}^*(p)}$. As $f_\mathbf{s}$ and $f_{\mathbf{s}^*(p)}$ are both within $o(n^{-\frac{1}{2}})$ of $\mathrm{OPT}_n$, $\Delta(f_\mathbf{s}) = 0$. By \Cref{thm:optimal_deterministic_appendix}, $\Delta(f_\mathbf{1}) = \Delta_p(\mathbf{1}) > 0$, since $\mathbf{1}$ is not a positive affine transformation of $\mathbf{s}^*(p)$. Hence $f_\mathbf{s}$ is asymptotically average-case optimal.

When $p = 1$, $\psi_1^{-1}(y) = 1 + y$ is affine, so the remainder $\rho_1$ in \Cref{lem:expansion} vanishes, $\varepsilon_n = 0$, and $f_\mathbf{s}$ maximizes $\mathrm{ESW}_1$ exactly for every $n$. In this case $G(\boldsymbol{\sigma}, a) = n\,\mathbb{E}[W^\text{PM}(\mathbf u, a, 1) \mid \boldsymbol{\sigma}] - n$, so $f^*$ is exactly the average-case optimal correspondence of \citet{boutilier2015optimal}.
\end{proof}

\subsection{What is the Optimal Rule?}
\label{app:subsec_optimal_rule_shape}

\begin{theorem}
\label{thm:optimal_deterministic_appendix}
Fix $p \in \mathbb{R}$, assume \Cref{a3} holds, and let $s_k^*(p) = \mathbb{E}[\psi_p(U_{(k)})]$ for $k = 1, \dots, m$. For every score vector $\mathbf{s} \in \mathbb{R}^m$ the limit
\[
\Delta_p(\mathbf{s}) = \lim_{n \to \infty} n^{\frac{1}{2}}\big(\mathrm{ESW}_p(f_{\mathbf{s}^*(p)}) - \mathrm{ESW}_p(f_\mathbf{s})\big)
\]
exists and is nonnegative, and $\Delta_p(\mathbf{s}) = 0$ if, and only if, $\mathbf{s} = \mathbf{s}^*(p)$. That is, up to positive affine transformation, $\mathbf{s}^*(p)$ is the unique asymptotically average-case optimal positional scoring rule.
\end{theorem}

The proof uses the following lemma, which holds exactly for every $n$. For $p = 1$ it is the result of \citet{boutilier2015optimal}, which shows that the Borda rule is optimal when utilities are uniform.

\begin{lemma}
\label{lemma:p0}
Fix $p \in \mathbb{R}$, assume \Cref{a3} holds, and let $V_p$ be as in \Cref{lem:expansion} and $\mathbf{s}^*(p)$ as in \Cref{thm:optimal_deterministic}. For every $n$, every preference profile $\boldsymbol{\sigma}$ and every $a \in A$,
\[
\mathbb{E}[V_p(\mathbf{u}, a) \mid \boldsymbol{\sigma}] = S_{\mathbf{s}^*(p)}(a, \boldsymbol{\sigma}).
\]
Consequently, for every (possibly randomized) social choice function $f$,
\begin{multline*}
\mathbb{E}[V_p(\mathbf{u}, f(\boldsymbol{\sigma}))] = \mathbb{E}[S_{\mathbf{s}^*(p)}(f(\boldsymbol{\sigma}), \boldsymbol{\sigma})]
\leq \mathbb{E}[V_p(\mathbf{u}, f_{\mathbf{s}^*(p)}(\boldsymbol{\sigma}))].
\end{multline*}
\end{lemma}

\begin{proof}
Fix voter $i$. By \Cref{a1}, relabelling the alternatives by any permutation $\pi$ does not change the distribution of voter $i$'s utility vector. Relabelling leaves the sorted utilities $(U_{(1)}, \dots, U_{(m)})$ unchanged and relabels the ranking $\sigma_i$, which is strict by \Cref{a2}. Hence, given the sorted utilities, every ranking is equally likely, so $\sigma_i$ is uniform and independent of $(U_{(1)}, \dots, U_{(m)})$. As $u_i(a) = U_{(\sigma_i(a))}$, conditioning on $\sigma_i$ only reveals which order statistic $u_i(a)$ is: if $\sigma_i(a) = k$ then $u_i(a)$ is distributed as $U_{(k)}$. Voters are independent (\Cref{a1}), so this also holds conditional on $\boldsymbol{\sigma}$, and the utilities of different voters remain independent given $\boldsymbol{\sigma}$. Therefore
\[
\mathbb{E}[V_p(\mathbf{u}, a) \mid \boldsymbol{\sigma}]
= \sum_{i=1}^n \mathbb{E}\big[\psi_p(U_{(\sigma_i(a))})\big]
= \sum_{i=1}^n s^*_{\sigma_i(a)}(p)
= S_{\mathbf{s}^*(p)}(a, \boldsymbol{\sigma}),
\]
where each term is finite by \Cref{a3}.

For the second claim, $f(\boldsymbol{\sigma})$ depends only on $\boldsymbol{\sigma}$ and on randomness independent of $\mathbf{u}$ given $\boldsymbol{\sigma}$. Applying the identity at $a = f(\boldsymbol{\sigma})$ and averaging gives $\mathbb{E}[V_p(\mathbf{u}, f(\boldsymbol{\sigma}))] = \mathbb{E}[S_{\mathbf{s}^*(p)}(f(\boldsymbol{\sigma}), \boldsymbol{\sigma})]$. This is at most $\mathbb{E}[\max_{a \in A} S_{\mathbf{s}^*(p)}(a, \boldsymbol{\sigma})]$, which $f_{\mathbf{s}^*(p)}$ attains.
\end{proof}

\begin{proof}[Proof of \Cref{thm:optimal_deterministic_appendix}]
For a preference profile $\boldsymbol{\sigma}$ of $n$ voters, let $\mathbf{Z}^{(n)}$ be the $m \times m$ matrix with entries
\[
\mathbf{Z}^{(n)}_{a,k} = n^{-\frac{1}{2}}\Big(\big|\{i : \sigma_i(a) = k\}\big| - \frac{n}{m}\Big),
\]
the centered and scaled number of voters ranking $a$ in position $k$. Every row and column of $\mathbf{Z}^{(n)}$ sums to $0$, so $\mathbf{Z}^{(n)}$ lies in the subspace $L$ of such matrices. For every score vector $\mathbf{s}$,
\begin{equation}
\label{eq:score_counts}
S_\mathbf{s}(a, \boldsymbol{\sigma}) = \frac{n}{m}\sum_{k} s_k + n^{\frac{1}{2}}\,(\mathbf{Z}^{(n)}\mathbf{s})_a,
\end{equation}
so the winners of $f_\mathbf{s}$ are $A_\mathbf{s}(\mathbf{Z}^{(n)})$, where $A_\mathbf{s}(\mathbf{Z}) = \operatorname*{arg\,max}_{a \in A}\, (\mathbf{Z}\mathbf{s})_a$.

Fix a non-constant \emph{reference vector} $\mathbf{r} \in \mathbb{R}^m$ and, for $\mathbf{Z} \in L$, let
\[
H_\mathbf{s}(\mathbf{Z}) = \max_{a \in A}\, (\mathbf{Z}\mathbf{r})_a - \frac{1}{|A_\mathbf{s}(\mathbf{Z})|}\sum_{a \in A_\mathbf{s}(\mathbf{Z})} (\mathbf{Z}\mathbf{r})_a \;\geq\; 0.
\]
By \eqref{eq:score_counts}, averaging over uniformly random tie-breaking,
\begin{equation}
\label{eq:score_loss}
\max_{a \in A} S_\mathbf{r}(a, \boldsymbol{\sigma}) - \mathbb{E}\big[S_\mathbf{r}(f_\mathbf{s}(\boldsymbol{\sigma}), \boldsymbol{\sigma}) \mid \boldsymbol{\sigma}\big] = n^{\frac{1}{2}}\, H_\mathbf{s}(\mathbf{Z}^{(n)}).
\end{equation}
So $H_\mathbf{s}$ measures how far the winner of $f_\mathbf{s}$ falls short of the highest $\mathbf{r}$-score, at scale $n^{\frac{1}{2}}$. Steps 1--3 hold for every non-constant $\mathbf{r}$. Step 4 takes $\mathbf{r} = \mathbf{s}^*(p)$.

\begin{enumerate}
    \item \textbf{Limit of the count matrix.} As shown in the proof of \Cref{lemma:p0}, each ranking $\sigma_i$ is uniform over all $m!$ rankings, and rankings are independent across voters (\Cref{a1}). So $\mathbf{Z}^{(n)}$ is a scaled sum of $n$ i.i.d.\ uniformly random permutation matrices, each centered by subtracting $\frac{1}{m}\mathbf{1}\mathbf{1}^\top$. By the multivariate Central Limit Theorem, $\mathbf{Z}^{(n)}$ converges in distribution to a centered Gaussian $\mathbf{Z}$ on $L$. Permutation matrices affinely span all matrices with row and column sums equal to $1$ (Birkhoff--von Neumann), so $\mathbf{Z}$ has a positive density on $L$.

    \item \textbf{Convergence of expectations.} If $\mathbf{s}$ is constant, then $A_\mathbf{s}(\mathbf{Z}) = A$ and $H_\mathbf{s}$ is continuous. Otherwise $H_\mathbf{s}$ is continuous wherever $A_\mathbf{s}(\mathbf{Z})$ is a single alternative. The remaining $\mathbf{Z}$, with $(\mathbf{Z}\mathbf{s})_a = (\mathbf{Z}\mathbf{s})_b$ for some $a \neq b$, form finitely many proper subspaces of $L$ and have probability $0$ under $\mathbf{Z}$. These subspaces are proper because, taking $s_k \neq s_l$, the matrix with row $a$ equal to $\mathbf{e}_k - \mathbf{e}_l$, row $b$ equal to $\mathbf{e}_l - \mathbf{e}_k$ and all other rows $0$ lies in $L$ and separates $a$ from $b$. By the continuous mapping theorem, $H_\mathbf{s}(\mathbf{Z}^{(n)}) \to H_\mathbf{s}(\mathbf{Z})$ in distribution. Moreover $H_\mathbf{s}(\mathbf{Z}^{(n)}) \leq 2 \max_{k}|r_k| \sum_{a,k}|\mathbf{Z}^{(n)}_{a,k}|$, whose second moment is bounded in $n$. So $(H_\mathbf{s}(\mathbf{Z}^{(n)}))_n$ is uniformly integrable and $\mathbb{E}[H_\mathbf{s}(\mathbf{Z}^{(n)})] \to \mathbb{E}[H_\mathbf{s}(\mathbf{Z})]$.

    \item \textbf{When is the limit zero?} We show $\mathbb{E}[H_\mathbf{s}(\mathbf{Z})] \geq 0$, with equality if, and only if, $\mathbf{s} = \alpha\,\mathbf{r} + \beta\mathbf{1}$ for some $\alpha > 0$ and $\beta \in \mathbb{R}$. If $\mathbf{s}$ is of this form, then $A_\mathbf{s} = A_\mathbf{r}$, so $H_\mathbf{s} = 0$. Otherwise we show $H_\mathbf{s} > 0$ with positive probability.
    \begin{itemize}
        \item If $\mathbf{s}$ is constant, then $A_\mathbf{s}(\mathbf{Z}) = A$. As the columns of $\mathbf{Z}$ sum to $0$, $\sum_{a}(\mathbf{Z}\mathbf{r})_a = 0$, so $H_\mathbf{s}(\mathbf{Z}) = \max_a (\mathbf{Z}\mathbf{r})_a$. This is positive unless $\mathbf{Z}\mathbf{r} = \mathbf{0}$, an event of probability $0$.
        \item Otherwise, the projections of $\mathbf{s}$ and $\mathbf{r}$ onto $\{\mathbf{d} : \sum_k d_k = 0\}$ are nonzero and not positive multiples of one another. So there is a $\mathbf{d}$ with $\sum_k d_k = 0$ and $\mathbf{s} \cdot \mathbf{d} < 0 < \mathbf{r} \cdot \mathbf{d}$. Fix $a \neq b$ and let $\mathbf{Z}_0 \in L$ have row $a$ equal to $\mathbf{d}$, row $b$ equal to $-\mathbf{d}$ and all other rows $0$. Then $A_\mathbf{r}(\mathbf{Z}_0) = \{a\}$ and $A_\mathbf{s}(\mathbf{Z}_0) = \{b\}$, so $H_\mathbf{s}(\mathbf{Z}_0) = 2\, \mathbf{r} \cdot \mathbf{d} > 0$. By continuity, $H_\mathbf{s} > 0$ on a neighbourhood of $\mathbf{Z}_0$ in $L$, which has positive probability under $\mathbf{Z}$.
    \end{itemize}

    \item \textbf{The limit $\Delta_p(\mathbf{s})$.} Take $\mathbf{r} = \mathbf{s}^*(p)$, which is not constant as $\psi_p$ is strictly increasing and $U_{(1)} > U_{(m)}$ by \Cref{a2}. By \Cref{lemma:p0} and \eqref{eq:score_loss},
    \[
    \mathbb{E}[V_p(\mathbf{u}, f_{\mathbf{s}^*(p)}(\boldsymbol{\sigma}))] - \mathbb{E}[V_p(\mathbf{u}, f_\mathbf{s}(\boldsymbol{\sigma}))] = n^{\frac{1}{2}}\,\mathbb{E}[H_\mathbf{s}(\mathbf{Z}^{(n)})].
    \]
    By \Cref{lem:expansion}, the difference in $\mathrm{ESW}_p$ is $\frac{\kappa_p}{n}$ times this difference, up to an error of $2\varepsilon_n = o(n^{-\frac{1}{2}})$. So by step 2,
    \[
    \Delta_p(\mathbf{s}) = \lim_{n\to\infty}\kappa_p\, \mathbb{E}[H_\mathbf{s}(\mathbf{Z}^{(n)})] = \kappa_p\, \mathbb{E}[H_\mathbf{s}(\mathbf{Z})].
    \]
    As $\kappa_p > 0$, step 3 shows this is nonnegative, with equality if, and only if, $\mathbf{s}$ is a positive affine transformation of $\mathbf{s}^*(p)$.

    \item \textbf{The ratio in \Cref{def:opt}.} For every social choice function $g$, \Cref{lemma:p0} gives $\mathbb{E}[V_p(\mathbf{u}, g(\boldsymbol{\sigma}))] \leq \mathbb{E}[V_p(\mathbf{u}, f_{\mathbf{s}^*(p)}(\boldsymbol{\sigma}))]$. By \Cref{lem:expansion} this gives $\mathrm{ESW}_p(g) - \mathrm{ESW}_p(f_{\mathbf{s}^*(p)}) \leq 2\varepsilon_n$, so $\mathrm{OPT}_n - \mathrm{ESW}_p(f_{\mathbf{s}^*(p)}) = o(n^{-\frac{1}{2}})$. We apply the argument at the start of \Cref{app:theorems} with $c_n = n^{\frac{1}{2}}$ and $f^\circ = f_{\mathbf{s}^*(p)}$. This gives $\Delta(f_\mathbf{s}) = \Delta_p(\mathbf{s})$, and $\Delta(f_\mathbf{1}) = \Delta_p(\mathbf{1}) > 0$ by step 4, as $\mathbf{1}$ is constant. So the ratio converges to $\Delta_p(\mathbf{s}) / \Delta_p(\mathbf{1})$, which is $0$ if, and only if, $\mathbf{s}$ is a positive affine transformation of $\mathbf{s}^*(p)$.
\end{enumerate}
\end{proof}

\clearpage


\subsection{Exact Egalitarian Social Welfare}
\label{app:egal_exact}

Throughout this section and the next we make a new assumption, slightly stronger than \Cref{a1}, which aligns our model with \citet{apesteguia2011justice}.


So each utility is an independent draw with CDF $F$, which is continuous by \Cref{a2}. Let $\ell = \inf \operatorname{supp}(F)$ denote the lowest possible utility. Write $\mathrm{ESW}_{-\infty}(f) = \mathbb{E}[\min_{i} u_i(f(\boldsymbol{\sigma}))]$ for the expected egalitarian social welfare of a social choice function $f$, and take $\mathrm{OPT}_n$ in \Cref{def:opt} with respect to $\mathrm{ESW}_{-\infty}$.

We also require that $F$ has a continuous and positive density near its lowest value. This holds, for example, for the uniform distribution on any interval and for the exponential distribution.

\begin{itemize}
    \asmitem{A4}{a4} For some $\delta > 0$, $F$ is continuously differentiable on $[\ell, \ell + \delta]$, with $\lambda := F'(\ell) > 0$.
\end{itemize}

\begin{theorem}
\label{thm:egal_sw_direct}
Assume \Cref{a1plus}, \Cref{a4}, and $\mathbb{E}[U] < \infty$, and denote the veto vector as $\mathbf{s}^\text{v} = (1, \dots, 1, 0)$. For every score vector $\mathbf{s} \in \mathbb{R}^m$ the limit
\[
\Delta_{-\infty}(\mathbf{s}) = \lim_{n \to \infty} n^{\frac32}\big(\mathrm{ESW}_{-\infty}(f_{\mathbf{s}^\text{v}}) - \mathrm{ESW}_{-\infty}(f_\mathbf{s})\big)
\]
exists and is nonnegative, and $\Delta_{-\infty}(\mathbf{s}) = 0$ if, and only if, $\mathbf{s} = \mathbf{s}^\text{v}$. That is, veto is the unique asymptotically average-case optimal positional scoring rule for egalitarian social welfare.
\end{theorem}

\begin{proof}
Fix a preference profile $\boldsymbol{\sigma}$. For each $a \in A$ let $K_a$ be the number of voters who rank $a$ last, $M_a = \min_i u_i(a)$ the egalitarian social welfare of $a$, and $L_a$ its last-place welfare (\Cref{lem:egal_bound}). Write $f_\text{v} = f_{\mathbf{s}^\text{v}}$, $K_{\min} = \min_a K_a$ and $b = f_\mathbf{s}(\boldsymbol{\sigma})$. As $S_{\mathbf{s}^\text{v}}(a, \boldsymbol{\sigma}) = n - K_a$, veto elects an alternative ranked last least often. The proof shows that egalitarian welfare is almost exactly last-place welfare, which veto maximizes, and that any other rule loses welfare in proportion to the extra last places of its winner.

\medskip \noindent
\textbf{Step 1: veto maximizes expected last-place welfare.}
By \Cref{lem:egal_bound}(b), $\mathbb{E}[L_a \mid \boldsymbol{\sigma}] = \varphi(K_a)$ decreases with the number of voters ranking $a$ last. Veto elects an alternative with the fewest such voters, so for any social choice function $g$,
\[
\mathbb{E}[L_{g(\boldsymbol{\sigma})} \mid \boldsymbol{\sigma}] \leq \varphi(K_{\min}) = \mathbb{E}[L_{f_\text{v}(\boldsymbol{\sigma})} \mid \boldsymbol{\sigma}].
\]

\medskip \noindent
\textbf{Step 2: last-place welfare approximates egalitarian welfare.}
Typically every alternative is ranked last by about $\frac{n}{m}$ voters. Let $B$ be the event that some alternative is ranked last by fewer than $\frac{n}{2m}$ voters. By Hoeffding's inequality, $B$ has exponentially small probability, and as expected welfare is bounded, profiles in $B$ can be ignored. Outside $B$, \Cref{lem:egal_bound}(c) shows that $L_a$ and $M_a$ differ in expectation by at most $C' n^{-2}$ for every alternative. Combined with Step 1, no social choice function has expected egalitarian welfare more than $O(n^{-2})$ above that of veto, that is,
\begin{equation}
\label{eq:egal_opt}
0 \leq \mathrm{OPT}_n - \mathrm{ESW}_{-\infty}(f_\text{v}) = O(n^{-2}).
\end{equation}
For a positional scoring rule $f_\mathbf{s}$, the same bounds give
\begin{multline}
\label{eq:egal_total}
n^{\frac32}\big(\mathrm{ESW}_{-\infty}(f_\text{v}) - \mathrm{ESW}_{-\infty}(f_\mathbf{s})\big) \\ = n^{\frac32}\,\mathbb{E}\big[(\varphi(K_{\min}) - \varphi(K_b))\mathbf{1}_{B^c}\big] + o(1),
\end{multline}
so it remains to compare expected last-place welfare.

\medskip \noindent
\textbf{Step 3: the cost of electing an alternative ranked last too often.}
Electing $b$ instead of the veto winner means electing an alternative ranked last $K_b - K_{\min}$ more times. Outside $B$, both are close to $\frac{n}{m}$, so by \Cref{lem:egal_gap} this costs
\[
\varphi(K_{\min}) - \varphi(K_b) \approx \frac{m\,(K_b - K_{\min})}{\lambda n^2}.
\]
The excess $K_b - K_{\min}$ is of order $n^{\frac12}$, the scale at which the counts $K_a$ fluctuate. Let $\mathbf{Z}^{(n)}$, $\mathbf{Z}$ and $H_\mathbf{s}$ be as in the proof of \Cref{thm:optimal_deterministic_appendix} with reference vector $\mathbf{r} = \mathbf{s}^\text{v}$. As $S_{\mathbf{s}^\text{v}}(a, \boldsymbol{\sigma}) = n - K_a$, \eqref{eq:score_loss} gives $\mathbb{E}[K_b - K_{\min} \mid \boldsymbol{\sigma}] = n^{\frac12} H_\mathbf{s}(\mathbf{Z}^{(n)})$. By steps 1 and 2 of that proof, $H_\mathbf{s}(\mathbf{Z}^{(n)})$ is uniformly integrable and $\mathbb{E}[H_\mathbf{s}(\mathbf{Z}^{(n)})] \to \mathbb{E}[H_\mathbf{s}(\mathbf{Z})]$. The approximations become exact in the limit, so by \eqref{eq:egal_total} the limit exists and
\[
\Delta_{-\infty}(\mathbf{s}) = \frac{m}{\lambda}\, \mathbb{E}[H_\mathbf{s}(\mathbf{Z})].
\]
So the welfare a positional scoring rule loses is proportional to how many extra last places, on average, its winner has.

\medskip \noindent
\textbf{Step 4: veto minimizes welfare loss.}
By step 3 of the proof of \Cref{thm:optimal_deterministic_appendix} with $\mathbf{r} = \mathbf{s}^\text{v}$, $\Delta_{-\infty}(\mathbf{s}) \geq 0$, with equality if, and only if, $\mathbf{s} = \alpha\,\mathbf{s}^\text{v} + \beta\mathbf{1}$ for some $\alpha > 0$ and $\beta \in \mathbb{R}$. In particular $\Delta_{-\infty}(\mathbf{1}) > 0$. By \eqref{eq:egal_opt}, $n^{\frac32}(\mathrm{OPT}_n - \mathrm{ESW}_{-\infty}(f_\text{v})) \to 0$, so the argument at the start of \Cref{app:theorems} applies with $c_n = n^{\frac32}$ and $f^\circ = f_\text{v}$. Hence veto is, up to positive affine transformation, the unique asymptotically average-case optimal positional scoring rule. The factor $\frac{m}{\lambda}$ cancels in the ratio, so its limit, $\mathbb{E}[H_\mathbf{s}(\mathbf{Z})] / \mathbb{E}[H_\mathbf{1}(\mathbf{Z})]$, depends only on $m$ and $\mathbf{s}$ and not on $F$.
\end{proof}

\begin{lemma}[Last-place welfare]
\label{lem:egal_bound}
Assume \Cref{a1plus}, \Cref{a4} and $\mathbb{E}[U] < \infty$. Fix a preference profile $\boldsymbol{\sigma}$ and $a \in A$. Let $K_a$ be the number of voters who rank $a$ last, $M_a := \min_{i} u_i(a)$, and let $L_a := \min_{i \,:\, \sigma_i(a) = m} u_i(a)$ be the \emph{last-place welfare} of $a$, with $L_a = +\infty$ if $K_a = 0$. Then
\begin{enumerate}
    \item[(a)] $M_a \le L_a$;
    \item[(b)] $\Pr(L_a > x \mid \boldsymbol{\sigma}) = (1-F(x))^{mK_a}$ for every $x$, so $\mathbb{E}[L_a \mid \boldsymbol{\sigma}] = \varphi(K_a)$, where $\varphi(k) := \int_0^\infty (1-F(x))^{mk}\,dx$ is nonincreasing in $k$; and
    \item[(c)] there is a constant $C'$, depending only on $F$ and $m$, such that $\mathbb{E}[L_a - M_a \mid \boldsymbol{\sigma}] \leq C' n^{-2}$ whenever $K_a \geq \frac{n}{2m}$.
\end{enumerate}
\end{lemma}

Intuitively, (b) says that the more voters rank $a$ last, the lower its expected last-place welfare. Part (c) says that the lowest utility given to $a$ almost always comes from a voter who ranks $a$ last.

\begin{proof}
Condition on $\boldsymbol{\sigma}$ throughout. Utilities of different voters remain independent. By \Cref{lemma:p0}, a voter who ranks $a$ in position $k$ gives $a$ a utility distributed as $U_{(k)}$, which by \Cref{a1plus} is the $k$-th largest of $m$ independent draws from $F$.

(a) The voters ranking $a$ last are a subset of all voters.

(b) $L_a$ is the smallest of $K_a$ independent copies of $U_{(m)}$, each the smallest of $m$ draws from $F$. So $L_a$ is the smallest of $mK_a$ independent draws from $F$, which gives the stated probability. Integrating over $x$ gives $\mathbb{E}[L_a \mid \boldsymbol{\sigma}] = \varphi(K_a)$. More voters ranking $a$ last means more draws, so $\varphi$ is nonincreasing.

(c) As $L_a - M_a = \int_0^\infty \mathbf{1}\{M_a \leq x < L_a\}\,dx$, we bound $\Pr(M_a \leq x < L_a \mid \boldsymbol{\sigma})$. This event requires two independent things. First, every voter ranking $a$ last gives $a$ more than $x$, which by (b) has probability $(1-F(x))^{mK_a}$. Second, some other voter gives $a$ at most $x$. Such a voter ranks some alternative below $a$, so at least two of their $m$ draws are at most $x$, which has probability at most $\binom{m}{2}F(x)^2$. A union bound over at most $n$ such voters gives
\[
\Pr(M_a \leq x < L_a \mid \boldsymbol{\sigma}) \leq \tbinom{m}{2}\, n\, (1-F(x))^{mK_a} F(x)^2.
\]
As $mK_a \geq \frac{n}{2}$, the factor $(1-F(x))^{mK_a}$ is exponentially small in $n$ except near $\ell$. Near $\ell$, $F$ has slope at least $\frac{\lambda}{2}$ by \Cref{a4}, so substituting $u = F(x)$ bounds the integral over $x$ by
\[
\frac{2}{\lambda}\int_0^1 (1-u)^{\frac{n}{2}} u^2\,du = O(n^{-3}).
\]
Multiplying by $\binom{m}{2} n$ gives (c).
\end{proof}

\begin{lemma}
\label{lem:egal_gap}
Assume \Cref{a1plus}, \Cref{a4} and $\mathbb{E}[U] < \infty$, and let $\varphi$ be as in \Cref{lem:egal_bound}(b). For every $\theta \in (0,1)$ there are constants $C > 0$ and $\gamma \in (0,1)$ such that, for all integers $1 \leq K \leq K'$,

\[
\big|\lambda\big(\varphi(K) - \varphi(K')\big) - A\big| \leq \theta A + C\gamma^K
\]

where

\[
A := \frac{1}{mK+1} - \frac{1}{mK'+1} = \frac{m(K' - K)}{(mK+1)(mK'+1)}.
\]

\end{lemma}

\begin{proof}
Write $\varphi(K) - \varphi(K') = \int_0^\infty g(F(x))\,dx$ with $g(u) = (1-u)^{mK} - (1-u)^{mK'} \geq 0$, which vanishes below $\ell$ as $g(0) = 0$. If $F$ were uniform with density $\lambda$ starting at $\ell$, substituting $u = F(x)$ would give exactly $\lambda(\varphi(K) - \varphi(K')) = \int_0^1 g(u)\,du = A$. We bound the two ways $F$ can deviate from this. By \Cref{a4}, choose $\varepsilon \in (0, \delta]$ such that $F'$ is within a factor $1 \pm \frac{\theta}{2}$ of $\lambda$ on $[\ell, \ell + \varepsilon]$, and let $\eta := F(\ell + \varepsilon) > 0$.
\begin{itemize}
    \item \emph{Near $\ell$.} Substituting $u = F(x)$ on $[\ell, \ell+\varepsilon]$ gives $\lambda$ times the integral there within a factor $1 \pm \theta$ of $\int_0^\eta g(u)\,du$. This differs from $A$ by at most $\int_\eta^1 (1-u)^{mK}\,du \leq (1-\eta)^{mK}$.
    \item \emph{Beyond $\ell + \varepsilon$.} Here $F(x) \geq \eta$, so $g(F(x)) \leq (1-F(x))^{mK} \leq (1-\eta)^{mK-1}(1-F(x))$, which integrates to at most $(1-\eta)^{mK-1}\,\mathbb{E}[U]$.
\end{itemize}
Both errors are at most a constant times $\gamma^K$ with $\gamma = (1-\eta)^m$, which gives the claim.
\end{proof}


\clearpage


\subsection{Exact Maximax Social Welfare}
\label{app:maximax_exact}

The theorem in this section mirrors that of \Cref{app:egal_exact} with the roles of the lowest and highest utilities, and of the last and first positions, exchanged. We again assume \Cref{a1plus}, and now introduce an assumption on the utility density near the upper end of the utility distribution.

Let $\bar\ell = \sup \operatorname{supp}(F)$ denote the highest possible utility. Write $\mathrm{ESW}_{+\infty}(f) = \mathbb{E}[\max_{i} u_i(f(\boldsymbol{\sigma}))]$ for the expected maximax social welfare of a social choice function $f$, and take $\mathrm{OPT}_n$ in \Cref{def:opt} with respect to $\mathrm{ESW}_{+\infty}$. 

Similar to \Cref{a4}, we require that $F$ has bounded support and a continuous density near its highest value which is positive at $\bar\ell$. This holds, for example, for the uniform distribution on any interval.

\begin{itemize}
    \asmitem{A5}{a5} $\bar\ell < \infty$, and for some $\delta > 0$, $F$ is continuously differentiable on $[\bar\ell - \delta, \bar\ell]$, with $\bar\lambda := F'(\bar\ell) > 0$.
\end{itemize}

As $F$ has bounded support, $\mathbb{E}[U] < \infty$ holds automatically.

\begin{theorem}
\label{thm:maximax_sw_direct}
Assume \Cref{a1plus} and \Cref{a5}, and denote the plurality vector as $\mathbf{s}^\text{p} = (1, 0, \dots, 0)$. For every score vector $\mathbf{s} \in \mathbb{R}^m$ the limit
\[
\Delta_{+\infty}(\mathbf{s}) = \lim_{n \to \infty} n^{\frac32}\big(\mathrm{ESW}_{+\infty}(f_{\mathbf{s}^\text{p}}) - \mathrm{ESW}_{+\infty}(f_\mathbf{s})\big)
\]
exists and is nonnegative, and $\Delta_{+\infty}(\mathbf{s}) = 0$ if, and only if, $\mathbf{s} = \mathbf{s}^\text{p}$. That is, plurality is the unique asymptotically average-case optimal positional scoring rule for maximax social welfare.
\end{theorem}

\begin{proof}
Fix a preference profile $\boldsymbol{\sigma}$. For each $a \in A$ let $\bar K_a$ be the number of voters who rank $a$ first, $\bar M_a = \max_i u_i(a)$ the maximax social welfare of $a$, and $\bar L_a$ its first-place welfare (\Cref{lem:malf_bound}). Write $f_\text{p} = f_{\mathbf{s}^\text{p}}$, $\bar K_{\max} = \max_a \bar K_a$ and $b = f_\mathbf{s}(\boldsymbol{\sigma})$. As $S_{\mathbf{s}^\text{p}}(a, \boldsymbol{\sigma}) = \bar K_a$, plurality elects an alternative ranked first most often. The proof mirrors that of \Cref{thm:egal_sw_direct}. Maximax welfare is almost exactly first-place welfare, which plurality maximizes, and any other rule loses welfare in proportion to the first places its winner lacks.

\medskip \noindent
\textbf{Step 1: plurality maximizes expected first-place welfare.}
By \Cref{lem:malf_bound}(b), $\mathbb{E}[\bar L_a \mid \boldsymbol{\sigma}] = \bar\varphi(\bar K_a)$ increases with the number of voters ranking $a$ first. Plurality elects an alternative with the most such voters, so for any social choice function $g$,
\[
\mathbb{E}[\bar L_{g(\boldsymbol{\sigma})} \mid \boldsymbol{\sigma}] \leq \bar\varphi(\bar K_{\max}) = \mathbb{E}[\bar L_{f_\text{p}(\boldsymbol{\sigma})} \mid \boldsymbol{\sigma}].
\]

\medskip \noindent
\textbf{Step 2: first-place welfare approximates maximax welfare.}
Typically every alternative is ranked first by about $\frac{n}{m}$ voters. Let $B$ be the event that some alternative is ranked first by fewer than $\frac{n}{2m}$ voters. By Hoeffding's inequality, $B$ has exponentially small probability, and as welfare is bounded by $\bar\ell$, profiles in $B$ can be ignored. Outside $B$, \Cref{lem:malf_bound}(c) shows that $\bar L_a$ and $\bar M_a$ differ in expectation by at most $C' n^{-2}$ for every alternative. Combined with Step 1, no social choice function has expected maximax welfare more than $O(n^{-2})$ above that of plurality, that is,
\begin{equation}
\label{eq:malf_opt}
0 \leq \mathrm{OPT}_n - \mathrm{ESW}_{+\infty}(f_\text{p}) = O(n^{-2}).
\end{equation}
For a positional scoring rule $f_\mathbf{s}$, the same bounds give
\begin{multline}
\label{eq:malf_total}
n^{\frac32}\big(\mathrm{ESW}_{+\infty}(f_\text{p}) - \mathrm{ESW}_{+\infty}(f_\mathbf{s})\big) \\ = n^{\frac32}\,\mathbb{E}\big[(\bar\varphi(\bar K_{\max}) - \bar\varphi(\bar K_b))\mathbf{1}_{B^c}\big] + o(1),
\end{multline}
so it remains to compare expected first-place welfare.

\medskip \noindent
\textbf{Step 3: the cost of electing an alternative ranked first too rarely.}
Electing $b$ instead of the plurality winner means electing an alternative ranked first $\bar K_{\max} - \bar K_b$ fewer times. Outside $B$, both are close to $\frac{n}{m}$, so by \Cref{lem:malf_gap} this costs
\[
\bar\varphi(\bar K_{\max}) - \bar\varphi(\bar K_b) \approx \frac{m\,(\bar K_{\max} - \bar K_b)}{\bar\lambda n^2}.
\]
The deficit $\bar K_{\max} - \bar K_b$ is of order $n^{\frac12}$, the scale at which the counts $\bar K_a$ fluctuate. Let $\mathbf{Z}^{(n)}$, $\mathbf{Z}$ and $H_\mathbf{s}$ be as in the proof of \Cref{thm:optimal_deterministic_appendix} with reference vector $\mathbf{r} = \mathbf{s}^\text{p}$. As $S_{\mathbf{s}^\text{p}}(a, \boldsymbol{\sigma}) = \bar K_a$, \eqref{eq:score_loss} gives $\mathbb{E}[\bar K_{\max} - \bar K_b \mid \boldsymbol{\sigma}] = n^{\frac12} H_\mathbf{s}(\mathbf{Z}^{(n)})$. By steps 1 and 2 of that proof, $H_\mathbf{s}(\mathbf{Z}^{(n)})$ is uniformly integrable and $\mathbb{E}[H_\mathbf{s}(\mathbf{Z}^{(n)})] \to \mathbb{E}[H_\mathbf{s}(\mathbf{Z})]$. The approximations become exact in the limit, so by \eqref{eq:malf_total} the limit exists and
\[
\Delta_{+\infty}(\mathbf{s}) = \frac{m}{\bar\lambda}\, \mathbb{E}[H_\mathbf{s}(\mathbf{Z})].
\]
So the welfare a positional scoring rule loses is proportional to how many fewer first places, on average, its winner has.

\medskip \noindent
\textbf{Step 4: plurality minimizes welfare loss.}
By step 3 of the proof of \Cref{thm:optimal_deterministic_appendix} with $\mathbf{r} = \mathbf{s}^\text{p}$, $\Delta_{+\infty}(\mathbf{s}) \geq 0$, with equality if, and only if, $\mathbf{s} = \alpha\,\mathbf{s}^\text{p} + \beta\mathbf{1}$ for some $\alpha > 0$ and $\beta \in \mathbb{R}$. In particular $\Delta_{+\infty}(\mathbf{1}) > 0$. By \eqref{eq:malf_opt}, $n^{\frac32}(\mathrm{OPT}_n - \mathrm{ESW}_{+\infty}(f_\text{p})) \to 0$, so the argument at the start of \Cref{app:theorems} applies with $c_n = n^{\frac32}$ and $f^\circ = f_\text{p}$. Hence plurality is, up to positive affine transformation, the unique asymptotically average-case optimal positional scoring rule. As in \Cref{thm:egal_sw_direct}, the factor $\frac{m}{\bar\lambda}$ cancels in the ratio, so its limit does not depend on $F$.
\end{proof}

\begin{lemma}[First-place welfare]
\label{lem:malf_bound}
Assume \Cref{a1plus} and \Cref{a5}. Fix a preference profile $\boldsymbol{\sigma}$ and $a \in A$. Let $\bar K_a$ be the number of voters who rank $a$ first, $\bar M_a := \max_{i} u_i(a)$, and let $\bar L_a := \max_{i \,:\, \sigma_i(a) = 1} u_i(a)$ be the \emph{first-place welfare} of $a$, with $\bar L_a = 0$ if $\bar K_a = 0$. Then
\begin{enumerate}
    \item[(a)] $\bar L_a \le \bar M_a$;
    \item[(b)] $\Pr(\bar L_a \leq x \mid \boldsymbol{\sigma}) = F(x)^{m\bar K_a}$ for every $x \geq 0$, so $\mathbb{E}[\bar L_a \mid \boldsymbol{\sigma}] = \bar\varphi(\bar K_a)$, where $\bar\varphi(k) := \bar\ell - \int_0^{\bar\ell} F(x)^{mk}\,dx$ is nondecreasing in $k$; and
    \item[(c)] there is a constant $C'$, depending only on $F$ and $m$, such that $\mathbb{E}[\bar M_a - \bar L_a \mid \boldsymbol{\sigma}] \leq C' n^{-2}$ whenever $\bar K_a \geq \frac{n}{2m}$.
\end{enumerate}
\end{lemma}

Intuitively, (b) says that the more voters rank $a$ first, the higher its expected first-place welfare. Part (c) says that the highest utility given to $a$ almost always comes from a voter who ranks $a$ first.

\begin{proof}
Condition on $\boldsymbol{\sigma}$ throughout. Utilities of different voters remain independent. By \Cref{lemma:p0}, a voter who ranks $a$ in position $k$ gives $a$ a utility distributed as $U_{(k)}$, which by \Cref{a1plus} is the $k$-th largest of $m$ independent draws from $F$.

(a) The voters ranking $a$ first are a subset of all voters, and utilities are positive.

(b) $\bar L_a$ is the largest of $\bar K_a$ independent copies of $U_{(1)}$, each the largest of $m$ draws from $F$. So $\bar L_a$ is the largest of $m\bar K_a$ independent draws from $F$, which gives the stated probability. As $0 \leq \bar L_a \leq \bar\ell$, integrating $\Pr(\bar L_a > x \mid \boldsymbol{\sigma})$ over $[0, \bar\ell]$ gives $\mathbb{E}[\bar L_a \mid \boldsymbol{\sigma}] = \bar\varphi(\bar K_a)$. More voters ranking $a$ first means more draws, so $\bar\varphi$ is nondecreasing.

(c) As $\bar M_a - \bar L_a = \int_0^{\bar\ell} \mathbf{1}\{\bar L_a \leq x < \bar M_a\}\,dx$, we bound $\Pr(\bar L_a \leq x < \bar M_a \mid \boldsymbol{\sigma})$. This event requires two independent things. First, every voter ranking $a$ first gives $a$ at most $x$, which by (b) has probability $F(x)^{m\bar K_a}$. Second, some other voter gives $a$ more than $x$. Such a voter ranks some alternative above $a$, so at least two of their $m$ draws exceed $x$, which has probability at most $\binom{m}{2}(1-F(x))^2$. A union bound over at most $n$ such voters gives
\[
\Pr(\bar L_a \leq x < \bar M_a \mid \boldsymbol{\sigma}) \leq \tbinom{m}{2}\, n\, F(x)^{m\bar K_a} (1-F(x))^2.
\]
As $m\bar K_a \geq \frac{n}{2}$, the factor $F(x)^{m\bar K_a}$ is exponentially small in $n$ except near $\bar\ell$. Near $\bar\ell$, $F$ has slope at least $\frac{\bar\lambda}{2}$ by \Cref{a5}, so substituting $u = 1 - F(x)$ bounds the integral over $x$ by
\[
\frac{2}{\bar\lambda}\int_0^1 (1-u)^{\frac{n}{2}} u^2\,du = O(n^{-3}).
\]
Multiplying by $\binom{m}{2} n$ gives (c).
\end{proof}

\begin{lemma}
\label{lem:malf_gap}
Assume \Cref{a1plus} and \Cref{a5}, and let $\bar\varphi$ be as in \Cref{lem:malf_bound}(b). For every $\theta \in (0,1)$ there are constants $C > 0$ and $\gamma \in (0,1)$ such that, for all integers $1 \leq K \leq K'$,

\[
\big|\bar\lambda\big(\bar\varphi(K') - \bar\varphi(K)\big) - A\big| \leq \theta A + C\gamma^K,
\]

where

\[
A := \frac{1}{mK+1} - \frac{1}{mK'+1} = \frac{m(K' - K)}{(mK+1)(mK'+1)}.
\]

\end{lemma}

\begin{proof}
Write $\bar\varphi(K') - \bar\varphi(K) = \int_0^{\bar\ell} g(1 - F(x))\,dx$ with $g(u) = (1-u)^{mK} - (1-u)^{mK'} \geq 0$. If $F$ were uniform with density $\bar\lambda$ ending at $\bar\ell$, substituting $u = 1 - F(x)$ would give exactly $\bar\lambda(\bar\varphi(K') - \bar\varphi(K)) = \int_0^1 g(u)\,du = A$. We bound the two ways $F$ can deviate from this. By \Cref{a5}, choose $\varepsilon \in (0, \delta]$ such that $F'$ is within a factor $1 \pm \frac{\theta}{2}$ of $\bar\lambda$ on $[\bar\ell - \varepsilon, \bar\ell]$, and let $\eta := 1 - F(\bar\ell - \varepsilon) > 0$.
\begin{itemize}
    \item \emph{Near $\bar\ell$.} Substituting $u = 1 - F(x)$ on $[\bar\ell - \varepsilon, \bar\ell]$ gives $\bar\lambda$ times the integral there within a factor $1 \pm \theta$ of $\int_0^\eta g(u)\,du$. This differs from $A$ by at most $\int_\eta^1 (1-u)^{mK}\,du \leq (1-\eta)^{mK}$.
    \item \emph{Below $\bar\ell - \varepsilon$.} Here $F(x) \leq 1 - \eta$, so $g(1 - F(x)) \leq F(x)^{mK} \leq (1-\eta)^{mK}$, which integrates to at most $\bar\ell\,(1-\eta)^{mK}$.
\end{itemize}
Both errors are at most a constant times $\gamma^K$ with $\gamma = (1-\eta)^m$, which gives the claim.
\end{proof}


\clearpage




\section{Additional Experiment Discussion}
\label{app:experiments}

Our experiments utilize \texttt{SciPy} \cite{2020SciPy-NMeth} for sampling from statistical distributions and \texttt{PrefSampling} \cite{boehmer2024guide} to generate Euclidean preferences. All experiments were performed on an M2 Macbook Air and the time required to run every experiment reported is roughly 7 hours (total; not 7 hours per experiment). This submission is accompanied by code with a detailed README that contains additional information about experiment runs.

\subsection{Preference Distributions}

\paragraph{Neutral Distributions}
We use six different neutral utility distributions in our experiments. We illustrate the PDF and CDF of each distribution in  \Cref{fig:utility_dist_pdf_cdf}. These are not the only distributions we could use, they are simply chosen to represent a variety of distributions. There is no reason to believe that human (or other agent) utility values would neatly correspond to any specific distribution so we focus on finding generalities that hold across distributions.

\begin{figure*}[tbp]
    \centering
    \includegraphics[width=\linewidth]{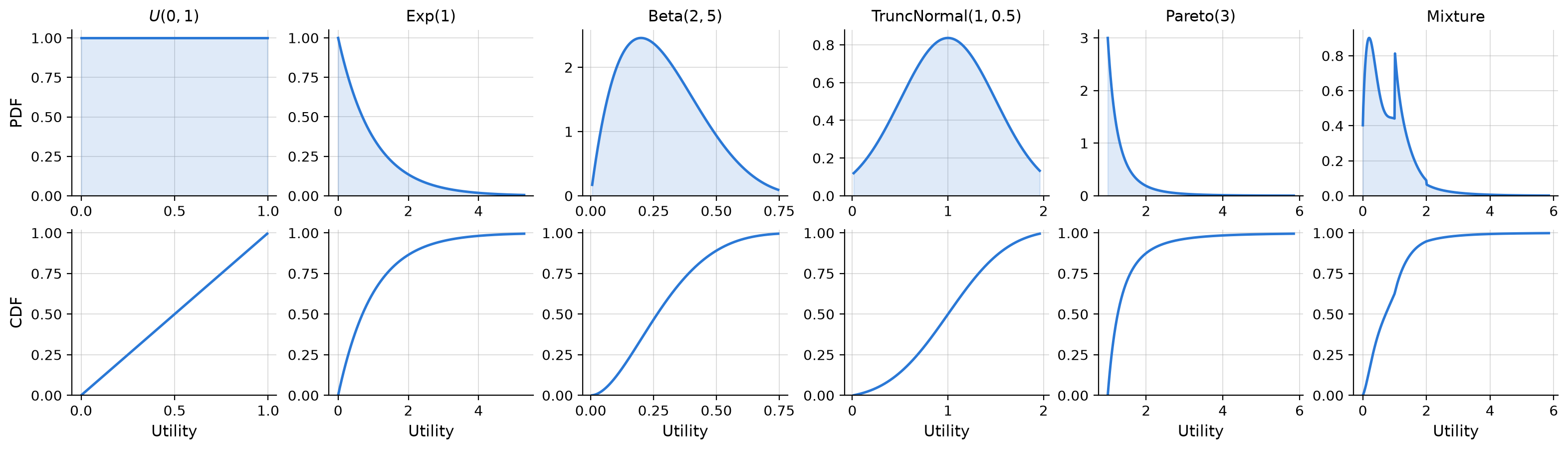}
    \caption{Probability Density Function and Cumulative Distribution Function for each utility distribution used in our experiments.}
    \label{fig:utility_dist_pdf_cdf}
\end{figure*}

\paragraph{Euclidean Distributions}

Euclidean distributions place both voters and alternatives into some underlying geometric space. We use cube and sphere topologies of varying dimensions, all with a diameter $1$. With a cube of three dimensions, voters and alternatives are placed within a cube at a location chosen uniformly at random. We use the \texttt{PrefSampling} library \cite{boehmer2024guide} to generate points within an $\ell$-dimensional cube or on the surface of an $\ell$-dimensional sphere. We illustrate a random sample of $m=20$ alternatives and $n=100$ voters in \Cref{fig:euclidean_structure}. Voter utility is proportional to the distance from the voter to each alternative with a voter at location $v_i$ receiving utility $u_i = 1 - \frac{\Vert v_i - a_j \Vert_2}{d}$ for an alternative at position $a_j$. Spheres have a diameter of $d=1$ and cubes have diameter $d = \sqrt{\ell}$.


\begin{figure*}[tbp]
    \centering
    \includegraphics[width=\linewidth]{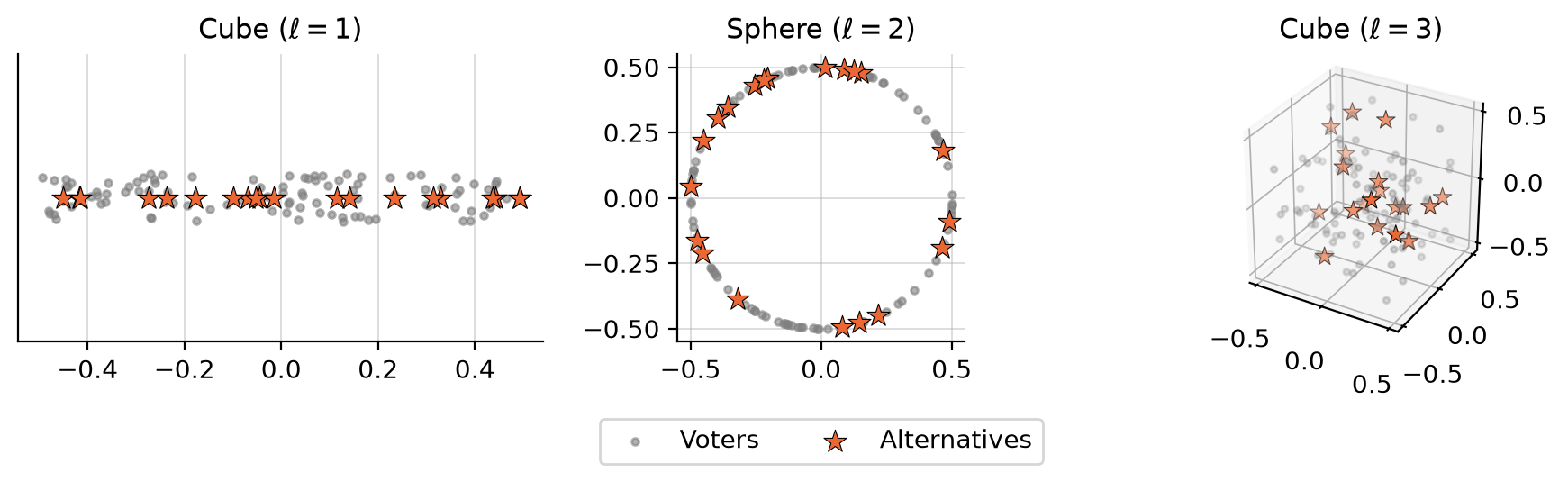}
    \caption{Example Euclidean distributions of varying dimensions and topologies. Note that height of voters in the left figure is only included for visualization and each item in the plot is defined fully by its position on the line.}
    \label{fig:euclidean_structure}
\end{figure*}

\paragraph{Plackett-Luce Distributions}

 The Plackett-Luce distribution is considered a generalization of the Bradley-Terry model of preferences and uses a vector of $m$ parameters which correspond to the quality of alternatives. An alternative with a higher quality is more likely to be given a higher utility value by a voter. We set weights for each alternative $j \in \{1, \dots, m\}$ as, 

 \[
 w_j = \exp{\gamma (\frac12 - \frac{j-1}{m-1})}
 \]

Where $\gamma$ parameterizes the family. Each voter then samples a utility value for each alternative from $Beta(w_j, 1)$. Using $\gamma = 0$ recreates a uniform distribution while higher $\gamma$ increases the expected range of utilities that a voter will take on. We show the CDF of each rank position for various $\gamma$ in \Cref{fig:pl_structure}.

\begin{figure*}[tbp]
    \centering
    \includegraphics[width=\linewidth]{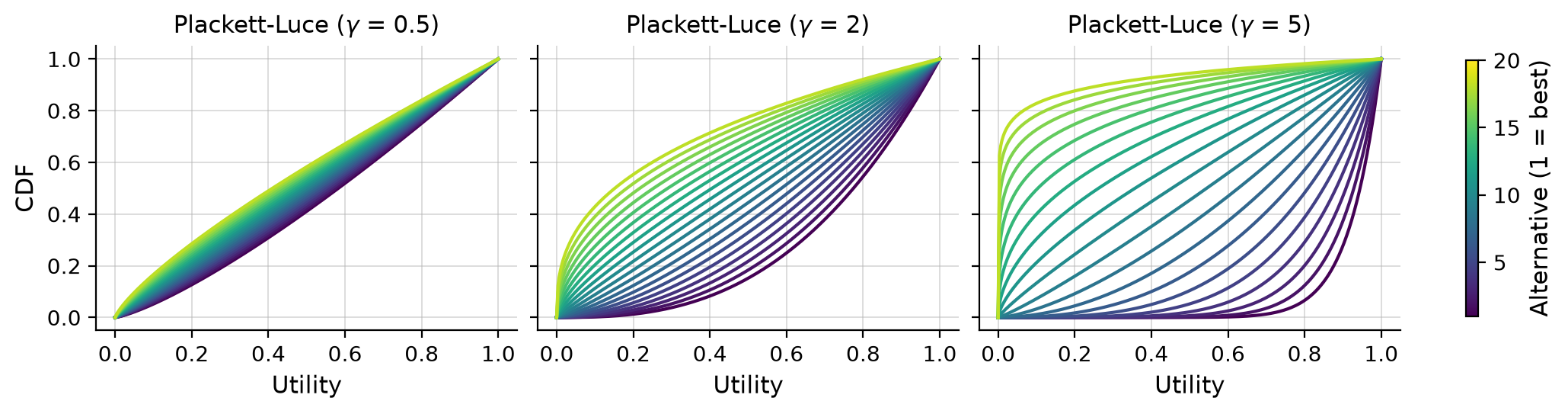 }
    \caption{CDF of Plackett-Luce utilities for various values of $\gamma$. Each series corresponds to the CDF of utility for a single rank position.}
    \label{fig:pl_structure}
\end{figure*}


\subsection{Structured Utility Distributions}

We include results on several additional parameterizations of distributions in which utilities are not sampled i.i.d. in \Cref{fig:opt_vec_nonneutral_all_dist,fig:welfare_loss_nonneutral_all_dist}. These highlight the novel dynamics shown in the paper body. The optimized score vector on Euclidean distributions in particular exhibit interesting structure which differs based on geometry.

\begin{figure*}[tbp]
    \centering
    \includegraphics[width=\linewidth]{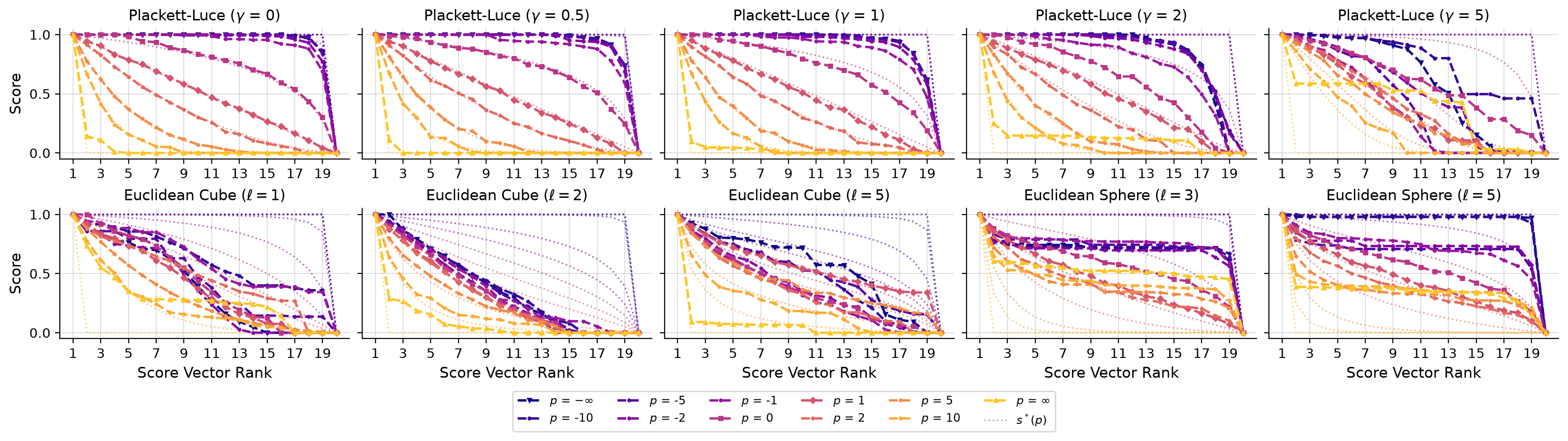}
    \caption{
    Optimized score vector across distributions as $p$ changes. Plackett-Luce with $\gamma = 0$ provides uniformly distributed utilities and shows how closely optimization matches the theoretical optimal value of $\mathbf{s}^*$.
    }
    \label{fig:opt_vec_nonneutral_all_dist}
\end{figure*}

\begin{figure*}[tbp]
    \centering
    \includegraphics[width=\linewidth]{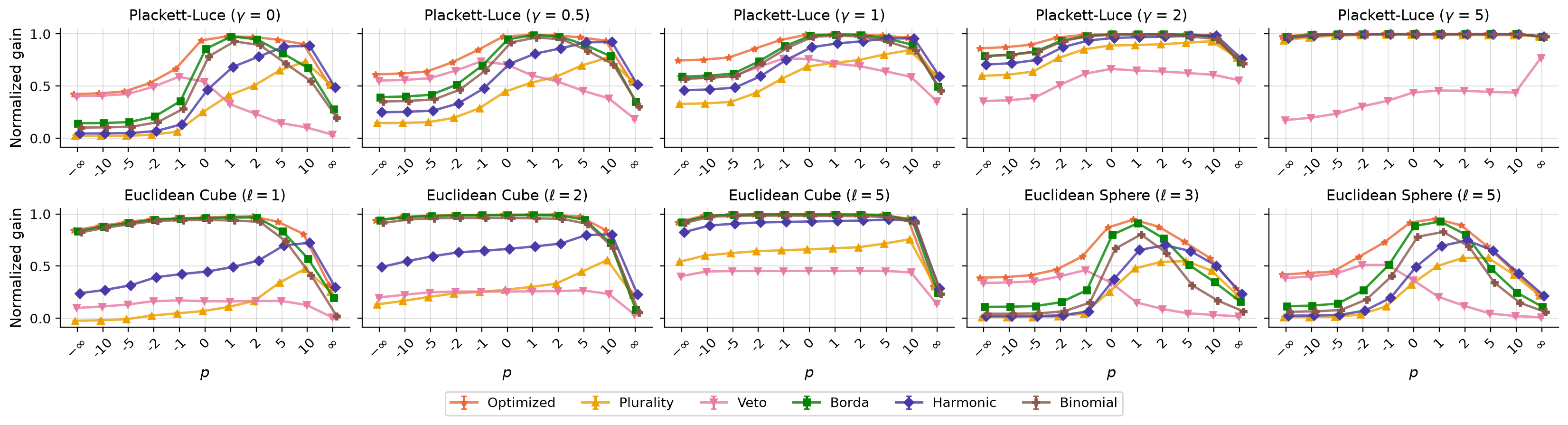}
    \caption{
    Normalized gain across additional distributions as $p$ changes. On Euclidean Cube distributions, Borda provides a very strong approximation of optimal welfare across multiple dimensions. As $\gamma$ increases, all rules but veto appear to perform near optimally with Plackett-Luce utilities.
    When $\gamma = 0$ utilities are distributed uniformly at random.
    }
    \label{fig:welfare_loss_nonneutral_all_dist}
\end{figure*}


\end{document}